%% file: hypergraph-clean-080921-arxiv.tex
\documentclass[onefignum,onetabnum]{siamart171218}

\input{shared}

\ifpdf
\hypersetup{
  pdftitle={A Bounded-Confidence Model on Hypergraphs},
  pdfauthor={A. Hickok, Y. Kureh, H. Z. Brooks, M. Feng, M. A. Porter}
}
\fi

\begin{document}

\maketitle

\begin{abstract}
People's opinions evolve over time as they interact with their friends, family, colleagues, and others. In the study of opinion dynamics on networks, one often encodes interactions between people in the form of dyadic relationships, but many social interactions in real life are polyadic (i.e., they involve three or more people). In this paper, we extend an asynchronous bounded-confidence model (BCM) on graphs, in which nodes are connected pairwise by edges, to {an asynchronous BCM on} hypergraphs, {in which {arbitrarily} many nodes can be connected by a single hyperedge}. We show that our hypergraph BCM converges to consensus under a wide range of initial conditions for the opinions of the nodes, including for non-uniform and asymmetric initial opinion distributions. We {also} show that, under suitable conditions, echo chambers can form on hypergraphs with community structure. We {demonstrate} that the opinions of individuals can sometimes jump from one opinion cluster to another in a single time step, a phenomenon (which we call ``opinion jumping'') that is not possible in standard dyadic BCMs. {Additionally, we} {observe} that there is a phase transition in the convergence time on {a complete hypergraph} when the variance $\sigma^2$ of the initial opinion distribution equals the confidence bound $c$. {We prove that the convergence time grows at least exponentially fast with the number of nodes when $\sigma^2 > c$ and the initial opinions are normally distributed.} Therefore, to determine the convergence properties of our hypergraph BCM when the variance and the number of hyperedges are both large, it is necessary to use analytical methods instead of relying only on Monte Carlo simulations.
\end{abstract}





\section{Introduction}

Social interactions with friends and acquaintances can {persuade} people to change their opinions about public figures \cite{kozitsin}, social issues \cite{siegel}, economic policy \cite{econ_apps}, and more. In opinion dynamics, researchers study how people's opinions about one or more topics evolve over time as they interact and influence each other \cite{galesic}. Traditionally, one models entities as nodes in a graph and {one models the social interactions of the entities} as edges that encode pairwise interactions between them \cite{bullo,newman2018}. The opinions of these entities can change as a result of such interactions. In the present paper, we build on these ideas by studying the effects of group interactions {on} opinion formation by modeling {these} interactions as hyperedges in a hypergraph. We find that such polyadic interactions play a key role in whether {or not a group reaches consensus} and {in} how long it takes to reach it.

We focus on \textit{continuous-opinion dynamics}, in which nodes have continuous-valued opinions. In our model, nodes hold opinions in $\mathbb{R}${. We} denote the \textit{opinion state} of the {system} by $\bm{x} \in \mathbb{R}^N$, where $N$ is the number of nodes. This is an appropriate model for opinions, such as the strength of support for a political candidate \cite{voteview}, that lie on a spectrum. By contrast, opinions such as whether one supports the Los Angeles Dodgers or {the} San Francisco Giants may leave little or no room for any middle ground. 

Bounded-confidence models {(BCMs) are models with continuous opinion states} in which individuals are influenced only by neighbors who hold opinions that are within some \textit{confidence bound} $c$ of their own opinion \cite{vixie, lorenz2007continuous}. Individuals who disagree with each other {too much} do not influence each other \cite{discrepency}. This models the concept of \emph{selective exposure} from social psychology; according to this principle, individuals tend to ignore information that is contrary to their current viewpoint \cite{sage, social_psych}. In traditional BCMs, each individual is a node in a graph and its neighbors are its adjacent nodes. A BCM is \textit{asynchronous} if only one pair of neighbors can interact at a time and is \textit{synchronous} if all pairs of neighboring nodes interact during each time step. The two most commonly studied BCMs are the (asynchronous) Deffuant--Weisbuch (DW) model~\cite{mixingBeliefs,meetDiscuss} and the (synchronous) Hegselmann--Krause (HK) model~\cite{hk}. See Ref.~\cite{lorenz2007continuous} for a review and a comparison of these two models, and see the introduction of Ref.~\cite{meng} for a recent summary of research on BCMs.

An important limitation of graphs is that they force one to consider only pairwise (i.e., ``dyadic'') interactions between nodes (as well as self-interactions, if one allows self-edges), whereas many social interactions involve many individuals at once {\cite{bick2021,petri_review, vision2020}.} One example of such a polyadic (i.e., ``higher-order'') social interaction is group messaging, such as group texting or e-mails with more than one recipient. We seek to examine the effects of polyadic interactions on opinion dynamics, so we develop and analyze an extension of BCMs to hypergraphs. In a hypergraph, a hyperedge can connect an arbitrary number of nodes to each other, rather than just two. In the context of opinion dynamics, one way to interpret such interactions is as a form of ``peer pressure'' \cite{synergy, melnik}, but other interpretations are also possible. Importantly, it is not possible to reduce the higher-order interactions in our hypergraph BCM to an aggregation of pairwise interactions.

Hypergraph extensions of opinion models have attracted much attention in the last few years \cite{hypergraph_bookchapter}. Sahasrabuddhe et al.~\cite{sahas} proposed a synchronous\footnote{We define a \emph{synchronous hypergraph model} to be one in which each node interacts with all {of} its {incident} hyperedges at once.} opinion model on hypergraphs, and {they conducted} numerical simulations of their model on {complete hypergraphs}, random hypergraphs, and real-world hypergraphs. Their model was not a {BCM}. By contrast, our model is a hypergraph extension of an asynchronous\footnote{We define an \emph{asynchronous hypergraph model} {as a model} in which only nodes in {a single hyperedge can interact in one time step}.} {BCM. We both conduct numerical simulations and derive analytical results about our model. Our hypergraph BCM produces qualitatively different dynamics from the model of Sahasrabuddhe et al. For example, consensus occurs under different circumstances in the two models. Additionally, in our model, the mean opinion is constant in time.} To the best of our knowledge, other hypergraph extensions of continuous-opinion models have only considered interactions with three or fewer nodes (i.e., hypergraphs in which the hyperedge ``sizes'' are no larger than $3$)\footnote{A few opinion models that consider large hyperedges have been developed for {frameworks} (such as voter models) with discrete-valued opinions \cite{majority_vote, coloring,petri2019}.} \cite{multibody, noonan}. In our hypergraph BCM, we find that ``large'' hyperedges (i.e., hyperedges that are incident to many nodes) are crucial for reaching consensus and that hypergraphs that have large hyperedges behave rather differently than hypergraphs with only small hyperedges.

A key {issue} in opinion dynamics {is} how model parameters, hypergraph structure (or the structure of other types of networks), and initial opinion states influence the opinion state to which a model converges. By applying the results of \cite{lorenz}, we show that the opinion state always converges to some \emph{limit state}. In the standard dyadic DW and HK models, the number of opinion clusters in the limit state depends on the confidence bound $c$. We say that the opinion state is at \textit{consensus} if there is a single opinion cluster; that is, every node has the same opinion. Most work on dyadic BCMs has drawn initial opinions uniformly at random from $[0, 1]$. In this case, the opinion state converges to consensus only when $c$ is above a certain threshold value \cite{mixingBeliefs, meetDiscuss, weisbuchClusters, universal, lorenz_thesis, hk_thresh}. By contrast, we show in \cref{equilibrium} that there is no such confidence-bound threshold for our hypergraph BCM. In \cref{sec:complete}, we prove this result for {complete hypergraphs}. In fact, we prove the following stronger statement: if the initial opinion distribution is bounded, then the opinion state on a complete hypergraph converges to consensus almost surely {if the hypergraph has} sufficiently many nodes. The case in which the initial opinion distribution is bounded includes {non-uniform probability distributions, asymmetric probability distributions, and probability distributions in which one draws initial opinions uniformly at random from an interval.} When one draws the initial opinions of the nodes in {a} complete hypergraph from a distribution with variance $\sigma^2 < c$, we prove that the probability of consensus approaches $1$ as the number of nodes approaches infinity. {For the special case in which the initial opinions are normally distributed, we also present numerical evidence that the opinion state converges to consensus even when $c \leq \sigma^2$.} We give a heuristic argument {to explain this observation.}

We also explore the behavior of our hypergraph BCM when the initial opinions depend on community structure \cite{comm_review, comm_fortunato}, in which dense sets of nodes are connected sparsely to other dense sets of nodes. A recent study of a dyadic BCM showed heuristically on ordinary graphs with community structure that separate opinion clusters tend to emerge for each community if the communities are not well-connected to each other \cite{fennell}. (See also Ref.~\cite{nonlinear_consensus}.) We use the term \textit{polarization} for this phenomenon\footnote{{Some researchers} refer to this phenomenon as ``opinion fragmentation'' and use the term ``polarized'' only when there are exactly two opinion {clusters.}}, and we use the term \textit{echo chambers} \cite{echo_flaxman, echo_fb, echo_political} for these different opinion clusters. In \cref{sec:comm}, we study our BCM on hypergraphs with community structure. We prove that polarization can occur if there is an upper bound on the sizes of the hyperedges that connect different communities. This provides a possible mechanism for the formation of echo chambers in hypergraphs. However, if there is no upper bound on the size of inter-community hyperedges and each community forms a complete hypergraph (i.e., a \emph{hyperclique}) and has sufficiently many nodes, we prove that our hypergraph BCM converges to consensus.

Using numerical simulations, we {demonstrate} that our {theoretical} results about reaching consensus are robust. The theorems in \cref{sec:complete} require that the hypergraph is complete, and some of the results in \cref{sec:comm} require that the communities form hypercliques. However, in practice, we can relax these conditions and the nodes' opinions still eventually reach consensus on the hypergraph in the former case and on individual communities in the latter case. In \cref{sparse}, we study the behavior of our BCM on sparse Erd\H{o}s--R\'enyi-like hypergraphs by using Monte Carlo simulations. In \cref{sec:enron}, we study the behavior of our model on the Enron e-mail hypergraph \cite{benson}, in which the nodes are Enron employees and hyperedges encode e-mails between these employees. Hypergraphs that one constructs from empirical data are interesting examples both because {typically they are} sparse and because their hyperedges are usually small in comparison to the number of nodes.

The convergence time of our a BCM is a significant factor to consider when we are running numerical simulations of it. In \cref{time}, we partially characterize the conditions under which our hypergraph BCM converges in finite time. In particular, we prove that it almost surely converges in finite time on {a} complete hypergraph. By comparison, the dyadic DW model for ordinary graphs usually does not converge in finite time, although the HK model always does \cite{hk_time}. We also {observe} that there is a phase transition in the convergence time of our BCM on {a complete hypergraph} when the variance $\sigma^2$ of the initial opinion distribution equals the confidence bound $c$. {We prove that the expected convergence time of our BCM grows at least exponentially fast with the number $N$ of nodes when $\sigma^2 > c$ and the initial opinions are normally distributed. When $\sigma^2 < c$, our numerical experiments on complete hypergraphs converge much faster than when $\sigma^2 > c$.} Meng et al. \cite{meng} demonstrated numerically that the standard dyadic DW model also has a phase transition in convergence time. It is important to understand this phase transition because when one is running {a Monte Carlo simulation} of a BCM, one chooses a finite cutoff time to stop the simulations. Without analysis of the convergence time, one may accidentally cut off the numerical simulations too early and mistakenly conclude that there is a phase transition in the limit state when what has actually occurred is a phase transition in convergence time.

When studying opinion dynamics, {it is also desirable} to understand the evolution of the opinion state before {reaching} a limit state. In \cref{sec:jump}, we investigate a phenomenon, which we call \textit{opinion jumping}, in which the opinion of a node changes by more than $c$ in a single time step. Opinion jumping allows nodes with extreme opinions to jump {close} to the mean of the opinion distribution in a single time step. This behavior cannot occur in the classical dyadic DW or HK models because nodes in those BCMs interact only with neighbors whose opinions are sufficiently similar to their own.

Our paper proceeds as follows. In \cref{sec:def}, we give a formal definition of our hypergraph BCM. In \cref{equilibrium}, we present our results about its limit state. These are the main results of our paper. In \cref{time}, we discuss convergence time. In \cref{sec:jump}, we examine opinion jumping and quantify how often it occurs. We conclude and discuss future work in \cref{conclusion}. Our code is available at \url{https://bitbucket.org/ahickok/hypergraph-bcm}.


\section{A Bounded-Confidence Model on Hypergraphs}\label{sec:def}

In this section, we develop an extension of the Deffuant--Weisbuch (DW) model to hypergraphs. We start by presenting the standard DW model on graphs.

In the standard dyadic DW model, opinion dynamics occur on an unweighted and undirected graph whose edges encode social ties. At each discrete time $t$, one chooses an edge $e = \{i, j\}$ uniformly at random. If the difference $|x_i(t)-x_j(t)|$ of opinions between nodes $i$ and $j$ is below some confidence bound $c_{i,j}$, then nodes $i$ and $j$ adjust their opinions as follows:
\begin{align}
    x_i(t+1) &= x_i(t) + m_{i,j} (x_j(t)-x_i(t))\,, \notag \\
    x_j(t+1) &= x_j(t) + {m_{j, i}} (x_i(t)-x_j(t))\,,
\end{align}
where $m_{i,j}$ is an element of the matrix of convergence parameters. Otherwise, the opinions of nodes $i$ and $j$ are too far apart at time $t$, so  $x_i(t+1)=x_i(t)$ and $x_j(t+1)=x_j(t)$. The opinions of all other nodes do not change when we do this update. With this type of update rule, the mean opinion of the nodes in a network is a conserved quantity. The confidence bounds $c_{i,j}\in[0,\infty)$ model {the level of open-mindedness of individuals} to the opinions of others \cite{petty1986elaboration}. The convergence parameters $m_{i,j}\in[0,0.5]$ (which resemble the trust parameters in DeGroot models \cite{vixie}) control the rate at which individuals adjust their opinions \cite{meetDiscuss, mixingBeliefs}. Using a single value $c$ and $m$ for all pairs leads to what is sometimes called the ``homogeneous'' DW model.

We now define our BCM on hypergraphs as an extension of the homogeneous DW model. A hypergraph is a generalization of a graph that allows {interactions} between arbitrarily many nodes. {That is, interactions between multiple nodes can be either dyadic or polyadic.} The space of possible hyperedges is the power set $\mathcal{P}(V)$ of the set $V$ of nodes. Let $N:= |V|$ denote the number of nodes. In the hypergraphs $(V,E)$ that we consider, we restrict the hyperedge set $E\subset\{e \in \mathcal{P}(V) |\; |e|\geq 2\}$ so that each hyperedge is incident to at least two nodes. Prohibiting hyperedges that are attached to only a single node (these are called ``self-hyperedges'') affects only {the} convergence time; it does not affect the limit state. All of our hypergraphs are unweighted and undirected. In our BCM, there is a time-dependent opinion state $\bm{x}(t)\in O^{N}$, where we take the opinion space $O$ to be the real line $\mathbb{R}$. We use $x_i(t)$ to denote the opinion of node $i$ at time $t$.

To generalize the notion of a confidence bound to hyperedges, we define a \emph{discordance function} $d : E\times O^{N}\rightarrow \mathbb{R}_{\geq0}$ that maps a hyperedge and opinion state to a real number. We use this function, which quantifies the level of disagreement among the nodes that are incident to a hyperedge, to determine whether or not these nodes update their opinions. {We consider the following family of discordance functions:}
\begin{equation}\label{eq:BCM_discordance}
    d_\alpha(e,\bm{x}) = \left( \frac{1}{|e|-1} \right)^\alpha \sum_{i \in e} \left(x_i -  \bar{x}_e\right)^2 \,, 
\end{equation}
{which is parametrized by the scalar $\alpha$, where $\bar{x}_e = (\sum_{i\in e} x_i)/|e|$.} If the discordance $d_{\alpha}(e, \bm{x}(t))$ is less than the confidence bound $c$, {we} say that the hyperedge $e$ is \emph{concordant} at time $t$. Otherwise, we say that $e$ is \emph{discordant}. 

{The choice $\alpha = 1$ is a noteworthy special case.} The function $d_1(e,\bm{x})$ is equal to the unbiased sample variance of the opinions of the nodes that are incident to $e$. {This models a situation in which nodes with moderate opinions can mediate between nodes in a group with extreme opinions as long as the overall disagreement within the group is not too high.} For example, let $\bm{x} = (0,1,0.5)$ and consider the hyperedges $e=\{1,2\}$ and $e'=\{1,2,3\}$. We see that $d_1(e,\bm{x}) = 0.5 > 0.25 = d_1(e',\bm{x})$, even though $e\subset e'$. One can interpret node $3$'s role in the interaction as that of a mediator who reduces the amount of discordance, thereby potentially yielding an update that otherwise {would} not occur. The scaling by $1/(|e|-1)$ in \eqref{eq:BCM_discordance} prevents {advantaging hyperedges with few nodes over hyperedges with many nodes when we update opinions.} Specifically, if the opinions are independent and identically distributed, then the expected $d_1$-discordance of any subset of nodes is
\begin{equation}
    \mathbb{E}[d_1(e,\bm{x})]=\mathbb{E}[d_1(e',\bm{x})] \; \; \text{for all}\,\, e,e'\in E \, .
\end{equation}
We set the discordance function to $d = d_1$ for the remainder of this paper.

Another {noteworthy} special case, which we do not {consider further} in the present paper, {is} $\alpha = 0$. The function $d_0(e,\bm{x})$ penalizes large hyperedges, in the sense that 
\begin{equation}\label{eq:hyperedge_monotonic}
    d(e,\bm{x}) \leq d(e',\bm{x}) \, \text{ if } \, e\subset e' \,,
\end{equation} 
with equality holding if and only if $x_i=\bar{x}_e$ for all $i\in e'\setminus e$. We use the term \emph{hyperedge monotonic} for a discordance function that satisfies \cref{eq:hyperedge_monotonic}. This models a situation in which large groups tend to be less effective than small groups at changing opinions.

We employ {asynchronous updates (as in the DW model).} At each discrete time, we randomly select a hyperedge $e$ from $E$ according to some probability distribution. For mathematical convenience, we use the uniform distribution over $E$. If the discordance $d(e,\bm{x})$ is less than the confidence bound $c$, the nodes $i\in e$ update their opinions $x_i$ to the mean opinion $\bar{x}_e$; otherwise, their opinions do not change. {One way to think of this update is that nodes $i \in e$ are {``peer-pressured''} into conforming to the mean opinion of the group when the overall discordance of the group is {sufficiently small}. More precisely,} if we select hyperedge $e$ at time $t$, {the} update rule for each node $i$ is
\begin{equation}\label{eq:BCM_hypergraph_update}
   	 x_i(t+1) = \begin{cases}
				\bar{x}_e(t) \, , &  \text{if }\, {i \in e \,\text{ and }}\, d(e,\bm{x}) <c\\
				x_i(t) \, , & \text{otherwise}\,.
\end{cases}
\end{equation}
The sequence $\bm{x}(0), \bm{x}(1), \bm{x}(2), \ldots$ of opinion states is a discrete-time Markov chain with a continuous state space.

{For the special case of a hypergraph that is a} graph (i.e., if $|e|=2$ for all $e\in E$), our generalized BCM reduces to a standard DW model with a rescaled confidence bound $c$. This rescaling arises from the difference in discordance functions: the standard DW model uses the absolute value $|x_i-x_j|$ of the difference of opinions, whereas our model uses $\frac{1}{2} (x_i-x_j)^2$ {[see \cref{eq:BCM_discordance}]}. On hypergraphs that are graphs, {our generalized BCM with confidence bound $\frac{1}{2}c^2$ is equivalent to the standard DW model with confidence bound $c$. Therefore,} our generalized BCM does not exactly reduce to the standard DW {model. However, the} two models are still easy to compare. The advantage of our choice is that {the discordance $d_1$} is equal to the unbiased sample {variance; this is helpful for deriving our analytical} results.


\section{The Limit State of Our Hypergraph BCM}\label{equilibrium}

We say that the opinion state \textit{converges} to $\bm{x}^*$ if $\lim_{t \to \infty} \bm{x}(t) = \bm{x}^*$. We refer to $\bm{x}^*$ as the \textit{limit state}. An \textit{opinion cluster} in the limit state is a collection of nodes that all have the same opinion in the limit state. The \textit{opinion value} of an opinion cluster is the opinion $\gamma \in \mathbb{R}$ such that $x^*_i = \gamma$ for all nodes $i$ in that cluster.

The opinion state converges to \textit{consensus} if there is a $\gamma \in \mathbb{R}$ such that $x^*_i = \gamma$ for all $i$. Equivalently, the opinion state converges to consensus if there is exactly one opinion cluster in the limit state. If the opinion state converges to consensus, it is necessarily true that $\gamma = \frac{1}{N}\sum_{i=1}^N x_i(0)$ because the mean opinion of the nodes is constant with respect to time.  

{An opinion state $\bm{x}$ is} an \textit{absorbing} state if for all $ e \in E$, either $d(e, \bm{x}) \geq c$ or $d(e, \bm{x}) = 0$ (i.e., $x_i = x_j$ for all $i, j \in e$). If $\bm{x}(T)$ is an absorbing state, then $\bm{x}(t) = \bm{x}(T)$ for all $t \geq T$. We will prove in \cref{no_updates} that the limit state is almost surely an absorbing state.

{We now show that for any} initial opinion state $\bm{x}(0)$, the opinion state of our hypergraph BCM converges in the limit $t \to \infty$.

\begin{theorem}\label{converge}

Let $\bm{x}(0)$ be the initial opinion {state and} update the opinion state $\bm{x}(t)$ {according to \cref{eq:BCM_hypergraph_update}}. It follows that the limit state $\bm{x}^* := \lim_{t \to \infty} \bm{x}(t)$ exists.
\end{theorem}
\begin{proof}Let $A(\bm{x}(t), t)$ be the $N \times N$ matrix such that $\bm{x}(t+1) = A(\bm{x}(t), t)\bm{x}(t)$, and let $e_t$ denote the hyperedge that we choose at discrete time $t$. If $e_t$ is discordant, then $A(\bm{x}(t), t) = I_N$. If $e_t$ is concordant, then $A(\bm{x}(t), t)$ is the matrix with entries
\begin{equation*}
	A(\bm{x}(t), t)_{ij} = \begin{cases} 1/|e_t| \,, & i, j \in e_t \\
\delta_{ij} \,, & \text{otherwise} \,. \end{cases}
\end{equation*}
The matrix $A(\bm{x}(t), t)$ satisfies the following conditions \cite{lorenz}:
\begin{enumerate}
    \item[(1)]\emph{Every agent has a bit of self-confidence}: The diagonal entries of $A(\bm{x}(t), t)$ are positive.
    \item[(2)]\emph{Confidence is mutual}: That is, for all {pairs $i, j \in \{1, \ldots, N\}$,} we have that $A(\bm{x}(t), t)_{ij} > 0$ if and only if $A(\bm{x}(t), t)_{ji} > 0$.
    \item[(3)]\emph{Positive weights do not converge to {$0$}}: There is a $\delta > 0$ such that every positive entry of $A(\bm{x}(t), t)$ is at least $\delta$. In our model, every positive entry is at least $1/N$.
\end{enumerate}

For any two times $t_0$ and $t_1$ with $t_0 < t_1$, Lorenz \cite{lorenz} defined the \textit{accumulation matrix} 
\begin{equation*}
	A(t_0, t_1) := A(\bm{x}(t_1 -1), t_1 -1)A(\bm{x}(t_1 -2), t_1 - 2)\times \cdots \times A(\bm{x}(t_0 +1),t_0 + 1)A(\bm{x}(t_0), t_0)\,.
\end{equation*}	 
Using this notation, $\bm{x}(t) = A(0, t)\bm{x}(0)$. He showed that if conditions (1)--(3) are satisfied, then there is a time $t_0$ and an ordering of the nodes such that 
\begin{equation}\label{Alim}
	\lim_{t \to \infty} A(0,t) = \begin{bmatrix} K_1 & & 0 \\
		&\ddots \\ 0 & & K_p \end{bmatrix} A(0, t_0) \,,
\end{equation}
where each $K_i$ is a row-stochastic matrix {whose rows are all the same}. The DeGroot model \cite{degroot}, the dyadic DW model \cite{mixingBeliefs}, and the HK model \cite{hk} all satisfy conditions (1)--(3).

Let $I_i$ be the set of nodes in the block $K_i$. \Cref{Alim} implies that $\bm{x}(t)$ converges to some opinion state $\bm{x}^*$ such that $x^*_j = x^*_k$ for all $j, k \in I_i$.
\end{proof}

{We will use the following lemma repeatedly in the subsections that follow.}

\begin{lemma}\label{no_updates}
{Let $\bm{x}(0)$ be the initial opinion {state,} and let $\bm{x}(t)$ be the opinion state {that is} determined by \cref{eq:BCM_hypergraph_update}. It follows that} the limit state $\bm{x}^* := \lim_{t \to \infty} \bm{x}(t)$ is almost surely an absorbing state.
\end{lemma}
\begin{proof}
By \cref{converge}, we know that 
$\bm{x}^*$ exists. If $\bm{x}^*$ is not an absorbing state, then there is a hyperedge $e \in E$ such that $d(e, \bm{x}^*) < c$ and $x_i^* \neq x_j^*$ for some $i, j \in e$. Let $\bar{x}_e^* = \frac{1}{|e|} \sum_{k \in e} x_k^*$, and note that $x_i^* \neq \bar{x}_e^*$. For all $\epsilon > 0$, there is a time $T$ such that 
\begin{align}
	|x_i(t+1) - x_i(t) | &< \epsilon \,, \label{eq:stable1}\\
	|x_i(t) - x_i^*| &< \epsilon \,, \label{eq:stable2}\\
	|d(e, \bm{x}(t)) - d(e, \bm{x}^*)| &< \epsilon \,, \label{eq:stable3}\\
	|\bar{x}_e(t) - \bar{x}_e^*| &< \epsilon  \label{eq:stable4}
\end{align}
for all $ t \geq T$. Choose $\epsilon < \min\{c - d(e, \bm{x}^*), |\bar{x}_e^* - x_i^*|/3\}$, and let $T$ be a time that satisfies \cref{eq:stable1}--\cref{eq:stable4}. With probability $1$, we choose every hyperedge in $E$ infinitely often (by the Borel--Cantelli lemma). Therefore, we choose $e$ at some time $t \geq T$ almost surely. If this happens, then $d(e, \bm{x}(t)) < c$ by \cref{eq:stable3} and we update the nodes of $e$ to obtain
\begin{equation*}
    |x_i(t+1) - x_i(t)| = |\bar{x}_e(t) - x_i(t)| > |\bar{x}_e^* - x_i^*| - 2 \epsilon > \epsilon
\end{equation*}
by \cref{eq:stable2} and \cref{eq:stable4}, contradicting \cref{eq:stable1}.
\end{proof}

\cref{no_updates} implies that there are almost surely no hyperedges that
are possible to update in the limit state.


\subsection{Our Hypergraph BCM on {Complete Hypergraphs}}\label{sec:complete}

In this subsection, we study the limit state of our hypergraph BCM on {complete hypergraphs}. On {a} complete hypergraph, every possible subset of nodes can interact with one another. Some of our results apply more generally to any hypergraph that includes the hyperedge $e = V$. We begin by presenting several lemmas that we then use to prove \cref{complete_variance}.

\begin{lemma}\label{edge_conc}
If the opinion distribution at time $t$ has finite variance $\sigma^2$, then
\begin{equation}\label{eq:limprobconc}
	\lim_{n \to \infty} \p[d(e, \bm{x}({t})) < c \mid |e| = n]=  \begin{cases}
1 \,, & c > \sigma^2 \\
\frac{1}{2} \,, & c = \sigma^2 \\
0 \,, & c < \sigma^2.
	\end{cases}
\end{equation}
\end{lemma}
\begin{proof}
The discordance of a hyperedge $e$ at time ${t}$ is the sample variance of the opinions $\{x_j({t}) \mid j \in e\}$. Let $s_n^2$ denote the sample variance of $n$ opinions. By {the} definition of the discordance function $d$, it follows that $\p[d(e, \bm{x}({t})) < c \mid |e| = n] = \p[s_n^2 < c]$. Because $\e[s_n^2] = \sigma^2$ and $\lim_{n \to \infty} \var[s_n^2] = 0$, Chebyshev's inequality implies that
\begin{align*}
	\lim_{n \to \infty} \p[s_n^2 < c]
&\geq \lim_{n \to \infty} 1 - \frac{\var[s_n^2]}{c - \sigma^2} = 1 \,, &\hspace{-1.5 cm}{c > \sigma^2}\,, \\
    \lim_{n \to \infty} \p[s_n^2 < c] 
    &\leq \lim_{n \to \infty} \frac{\var[s_n^2]}{\sigma^2 - c} = 0\,, &\hspace{-1.5 cm}{c < \sigma^2} \,.
\end{align*}
Note that $s_n^2$ converges asymptotically to the normal distribution $\mathcal{N}(\sigma^2, \sigma^2 (\kappa -1)/n)$, where $\kappa$ is the kurtosis of the initial opinion distribution. Because a normal distribution is symmetric, $\lim_{n \to \infty} \p[s_n^2 < c ] = \lim_{n \to \infty} \p[s_n^2 < \e[s_n^2]] = \frac{1}{2}$ if $c = \sigma^2$.
\end{proof}

The following lemma says that if a hyperedge $e$ has a nontrivial update at time $t$, then the discordance of each hyperedge $e'\supset e$ decreases [i.e., $d(e',\bm{x}(t+1)) < d(e',\bm{x}(t))$]. As a direct consequence of \cref{lemma:variance_decrease}, the discordance of the hyperedge $e = V$ is nonincreasing as the system evolves.

\begin{lemma}\label{lemma:variance_decrease}
Let $A = \{x_1, x_2, \ldots, x_n\}$ be a collection of $n$ real numbers, and let $A'=\{x_{i_1}, x_{i_2}, \ldots, x_{i_\ell}\}$ be some subcollection of $A$. Construct a new collection $B$ by taking the union of $ A\setminus A'$ and $\ell$ copies of the mean $\bar{x}_{A'}= \frac{1}{\ell} (\sum_{j=1}^\ell x_{i_j})$ of $A'$. The sample variances satisfy $s^2(B)\leq s^2(A)$, where equality holds if and only if $A=B$. 
\end{lemma}
\begin{proof}
The collections $A$ and $B$ have the same mean $\bar{x}_A$. {The following equality holds:}
\begin{equation}\label{RHS}
    (n-1)(s^2(A)-s^2(B)) = \sum_{j=1}^\ell (x_{i_j} - \bar{x})^2 - \ell ( \bar{x}_{A'} - \bar{x})^2\,.
\end{equation}
{We expand} the second term of the right-hand side of \eqref{RHS} {and write}
\begin{equation*}
    \sum_{j=1}^\ell (x_{i_j} - \bar{x})^2 - \frac{1}{\ell} \left( \sum_{j=1}^\ell (x_{i_j} - \bar{x})\right)^2\,.
\end{equation*}
We {then} define $y_j := x_{i_j} - \bar{x}$ to simplify the notation and write
\begin{equation*}
	\sum_{j=1}^\ell y_j^2 - \frac{1}{\ell} \left( \sum_{j=1}^\ell y_j\right)^2\,.
\end{equation*} 
Expanding the second term of the right-hand side of \eqref{RHS} further and simplifying yields
\begin{equation*}
	\begin{split}
\frac{1}{\ell} \left( (\ell - 1) \sum_{j=1}^\ell y_j^2 - 2\sum_{j=1}^\ell\sum_{k= j+1}^\ell y_j y_k \right)
=\frac{1}{\ell} \left(  \sum_{j=1}^\ell\sum_{k= j+1}^\ell (y_j- y_k)^2 \right) \geq 0 \, .
	\end{split}
\end{equation*}
Equality occurs if and only if $y_1= {\cdots} = y_\ell$, which proves the lemma.
\end{proof}

\begin{theorem}\label{complete_variance}
Suppose that $H$ is an $N$-node hypergraph that includes the hyperedge $e_N = V$. {Let $\bm{x}(0)$ be the initial opinion state, with opinions drawn independently from a distribution with a variance of $\sigma^2 < c$, and let $\bm{x}(t)$ be the opinion state {that is} determined by \cref{eq:BCM_hypergraph_update}.} It then follows that the probability of reaching consensus approaches $1$ as $N \to \infty$.
\end{theorem}

\begin{proof}
By \cref{lemma:variance_decrease}, the discordance function $d( e_N, \bm{x}(t))$ {is nonincreasing} in time. Therefore, if $e_N$ is concordant at time $0$, it is concordant for all $ t$. Additionally, if $e_N$ is concordant, $H$ converges to consensus the first time that one selects the hyperedge $e_N$. With probability $1$, this selection occurs at some finite time. This shows that
\begin{align*}
	\p[\text{consensus}] &\geq \p[\text{consensus} \mid d(e_N, \bm{x}(0)) < c]\p[d(e_N, \bm{x}(0)) < c]\\
&= \p[d(e_N, \bm{x}(0)) < c] \,.
\end{align*}
By \cref{edge_conc}, $\lim_{N \to \infty}\p[\text{consensus}] = 1$.
\end{proof}
\begin{remark}
In particular, \cref{complete_variance} applies to {a} complete hypergraph. 
\end{remark}

\begin{remark}
{We emphasize that} the probability distribution from which we draw the initial opinions need not be uniform or symmetric. The only condition on the distribution is that $\sigma^2 < c$.
\end{remark}


\subsubsection{Bounded Initial Opinions}

In this subsubsection, we assume that 
the probability distribution from which we draw initial opinions is supported on a bounded interval $[a, b]$. We present results about the limit state for this important case, which includes drawing the initial opinions uniformly at random from $[0,1]$ (a focal example of much prior work on the standard dyadic DW model \cite{lorenz2007continuous, mixingBeliefs, meetDiscuss}) as a special case. {In} Theorem~\ref{complete_converge}, {which is} the main result of this subsubsection, the probability distribution from which we draw the initial opinions need not be uniform or {symmetric; it only needs to be bounded.}

In the standard dyadic DW model, it has been observed using Monte Carlo simulations {for a} complete graph (with dyadic interactions only) that when the initial opinion distribution is $\mathcal{U}(0,1)$, there is a threshold confidence bound $c^* \approx \frac{1}{2}$ such that (1) the system converges to consensus with high probability for $c > c^*$ and (2) the system converges to approximately $\lfloor \frac{1}{2c} \rfloor$ opinion clusters for $c < c^*$ \cite{mixingBeliefs, meetDiscuss, weisbuchClusters}. A consensus threshold also exists for the standard dyadic HK model, but it occurs at a smaller confidence bound (of about $c^* = 0.19$) \cite{lorenz_thesis, hk_thresh}. In \cref{complete_converge}, we prove that no such threshold exists for our hypergraph BCM and that the opinion state converges to consensus almost surely {for sufficiently large} $N$ whenever the initial opinion distribution is bounded.

\begin{lemma}\label{compact}
Suppose that $H$ is the complete hypergraph with $N$ nodes and that $c \neq 0$. Let $\bm{x}(0)$ be the initial opinion state, where {we draw} $x_i(0)$ {for each $i$} from a bounded distribution {that is} supported on $[a, b]${,} and let $\bm{x}(t)$ be the opinion state that is determined by \cref{eq:BCM_hypergraph_update}. It is then the case that the number of opinion clusters in the limit state is almost surely less than or equal to $\frac{b-a}{\sqrt{2c}} + 1$.
\end{lemma}

\begin{proof}
Let $\bm{x}^* = \lim_{t \to \infty} \bm{x}(t)$ be the limit state, and let $\gamma_1, \ldots, \gamma_m \in \mathbb{R}$ be the opinion values of the $m$ opinion clusters. Let $\psi: V \to \{1, \ldots, m\}$ map {each node to its associated opinion cluster,} such that $x_i^* = \gamma_{\psi(i)}$. By \cref{no_updates}, it suffices to show that $m > \frac{b-a}{\sqrt{2c}}$ implies that $\bm{x}^*$ is not an absorbing state.

It must be the case that $\gamma_i \in [a,b]$ for all $i$ because $x_j(t) \in [a,b]$ for all $t$ and all $j$. If $m > \frac{b-a}{\sqrt{2c}}$, then there is a pair $\gamma_i, \gamma_j$ such that $\gamma_i \neq \gamma_j$ and $|\gamma_i - \gamma_j| < \sqrt{2c}$. Let $k_i$ be a node in the $i$th opinion cluster $\psi^{-1}(i)$, let $k_j$ be a node in the $j$th opinion cluster $\psi^{-1}(j)$, and let $e = \{k_i, k_j\}$. The limit state $\bm{x}^*$ is not an absorbing state because
\begin{align*}
    0 < d(e, \bm{x}^*) =  \frac{(\gamma_i - \gamma_j)^2}{2} < c \,.
\end{align*}
\end{proof}

{The} following lemma says that as $N \to \infty$, {it is almost surely the case that} at least one of the following three outcomes occurs: (1) the opinion state converges to consensus, (2) the number of opinion clusters approaches infinity, or (3) the difference between the opinion values of different opinion clusters approaches infinity.

\begin{lemma}\label{mclusters_unstable}
Suppose that $H$ is the complete hypergraph with $N$ nodes and that $c \neq 0$. {Let $\bm{x}(0)$ be the initial opinion state, and let $\bm{x}(t)$ be the opinion state that is determined by \cref{eq:BCM_hypergraph_update}.} Let $\bm{x}^*$ be the limit state, which exists by \cref{converge}, and let $\{\gamma_1, \ldots, \gamma_m\}$ be the opinion values of the opinion clusters. It then follows that the number $m$ of opinion clusters almost surely is either $m=1$ ({i.e.,} consensus) or satisfies
\begin{equation*}
    m \times \Big( \max_{k, j} \frac{(\gamma_k - \gamma_j)^2}{c} - 1\Big) \geq N\,.
\end{equation*}
\end{lemma}
\begin{proof}
By \cref{no_updates}, it suffices to show that if $N > m \times \Big( \max_{k, j} \frac{(\gamma_k - \gamma_j)^2}{c} - 1\Big)$ and $m>1$, then $\bm{x}^*$ is not an absorbing state. Let $\psi: V \to \{1, \ldots, m\}$ map {each node to its associated opinion cluster,} such that $x_i^* = \gamma_{\psi(i)}$, and let $N_i = |\psi^{-1}(i)|$ be the size of the $i$th opinion cluster. If $N > m \times \Big( \max_{k, j} \frac{(\gamma_k - \gamma_j)^2}{c} - 1\Big)$, then there exists an opinion cluster $i$ such that $ N_i > \max_{k, j} \frac{(\gamma_k - \gamma_j)^2}{c} - 1$. If $m>1$, there is {an opinion} $\gamma_j$ such that $\gamma_j \neq \gamma_i$. Let $e$ be a hyperedge of size $N_i + 1$ that is incident to the $N_i$ nodes in $\psi^{-1}(i)$ and $1$ node in $\psi^{-1}(j)$. The limit-state sample mean of the nodes that are incident to $e$ is
\begin{equation*}
    \bar{x}_e^*(t) = \frac{N_i}{N_i + 1}\gamma_i + \frac{1}{N_i +1} \gamma_j\,.
\end{equation*}
The limit state $\bm{x}^*$ is not an absorbing state because
 \begin{equation*}
	 0 < d(e, \bm{x}^*) = \frac{N_i (\gamma_i - \bar{x}_e^*(t))^2 + (\gamma_j - \bar{x}_e^*(t))^2}{N_i+1} = \frac{(\gamma_i - \gamma_j)^2}{N_i + 1} < c \,.
 \end{equation*}
\end{proof}

\begin{remark}
We only apply \cref{mclusters_unstable} to the case in which the initial opinion distribution is bounded, but it holds for any initial opinion distribution.
\end{remark}

{The} following theorem says that if the initial opinion distribution is bounded, then the system reaches consensus almost surely for {sufficiently large} $N$.

\begin{theorem}\label{complete_converge}
Suppose that $H$ is the complete hypergraph with $N$ nodes and that $c \neq 0$. Let $\bm{x}(0)$ be the initial opinion {state, where we draw $x_i(0)$ for each $i$ from a bounded distribution that is supported on $[a, b]$,} and let $\bm{x}(t)$ be the opinion state that is determined by \cref{eq:BCM_hypergraph_update}. If $N > (\frac{b-a}{\sqrt{c}} +1)( \frac{(b-a)^2}{c} - 1)$, then the opinion state converges to consensus almost surely.
\end{theorem}

\begin{proof}
Let $\bm{x}^*$ be the limit state, which exists by \cref{converge}, and let $\{\gamma_1, \ldots, \gamma_m\}$ be the opinion values of the opinion clusters. By \cref{compact}, $m \leq \frac{b-a}{\sqrt{c}} +1$ almost surely. It is necessarily true that $\gamma_i \in [a,b]$ for all $ i$, so $\max_{k, j} (\gamma_k - \gamma_j)^2 < (b-a)^2$. By \cref{mclusters_unstable}, $m = 1$ almost surely if $N > (\frac{b-a}{\sqrt{c}} +1)( \frac{(b-a)^2}{c} - 1)$.
\end{proof}


\subsubsection{Normally-Distributed Initial Opinions}

Assume that the initial opinions are normally distributed with mean $\mu$ and variance $\sigma^2$. When $\sigma^2 < c$, \cref{complete_variance} implies that the probability of consensus for the {$N$-node} complete hypergraph approaches $1$ as $N \to \infty$. Based on numerical evidence, we conjecture that the probability of reaching consensus for the {$N$-node} complete hypergraph approaches $1$ as $N \to \infty$ even when $\sigma^2 \geq c$, unless $c = 0$. Because the hypergraph is {complete,} \cref{mclusters_unstable} also applies.

In \cref{complete_normal}, we show a typical simulation when $\sigma^2 > c$. We have reduced the number of time steps by requiring that the hyperedge that we select at time $0$ is concordant. This requirement has no effect on the subsequent behavior or on the system's limit, {but it significantly reduces the number of time steps; we will see why this is true in our proof of \cref{normal_time}.} Observe that the opinion state converges to consensus. In $1000$ trials of our BCM for a confidence bound of $c = 1$, a complete hypergraph with $N = 200$ nodes, and an initial opinion distribution with standard deviation $\sigma = 1.2$, we find that the opinion state converges to consensus in every trial.

\begin{figure}
    \centering
    \includegraphics[width = .6\textwidth]{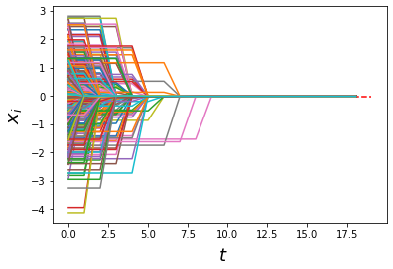}
    \caption{A typical simulation of our hypergraph BCM on the complete hypergraph with $N = 500$ nodes. Each curve traces the evolution of one node's opinion. We draw the initial opinions from $\mathcal{N}(0, \sigma^2)$ with $\sigma = 1.2$. We set the confidence bound to $c = 1$. We have reduced the number of time steps by requiring that the hyperedge that we choose at time $0$ is concordant. The opinion state converges to consensus.}
    \label{complete_normal}
\end{figure}

{Based on the results of our Monte Carlo simulations, we} conjecture that when $c \neq 0$ and the initial opinion distribution is normal, the probability of consensus approaches $1$ even when $\sigma^2 > c$. We now provide a heuristic explanation of this conjecture, although we do not have a mathematically rigorous proof of it.

We fix the variance to be $\sigma^2 > c$. At each {time step,} we select a hyperedge uniformly at random. Let ${e^*}$ be the first concordant hyperedge that we choose. In \cref{time}, we will show that if $e$ is an arbitrary hyperedge, then $\p[d(e, \bm{x}(0)) < c \mid |e| = n] \leq r^{n-1}$, where $r = e^{\frac{1}{2}(1 - c/\sigma^2)}\sqrt{(c/\sigma^2)}  < 1$. If it is true (and we suspect that it is) that $\p[d(e, \bm{x}(0)) < c \mid |e| = n] = ar^n$ for some constants $a$ and $r < 1$, then one can calculate that
\begin{equation*}
	\e[|e^*|] = \frac{(N (1+r)^{N-1} - N - 1}{\Big(\frac{1}{r} (1+r)^N - N - \frac{1}{r}\Big)} \sim \frac{r}{r+1}N \quad \text{as }\, N \rightarrow \infty\,. 
\end{equation*}	

That is, if we assume that $\p[d(e, \bm{x}(0)) < c \mid |e| = n] = ar^n$, then the expected size of $e^*$ grows linearly with $N$. In \cref{exp_conc_edge_size}, we show the results of Monte Carlo simulations to estimate $\e[|e^*|]$ as a function of $N$ without using the assumption $\p[d(e, \bm{x}(0)) < c \mid |e| = n] = ar^n$. As hypothesized, we observe a linear relationship.

\begin{figure}
    \centering
    \includegraphics[width = .5\textwidth]{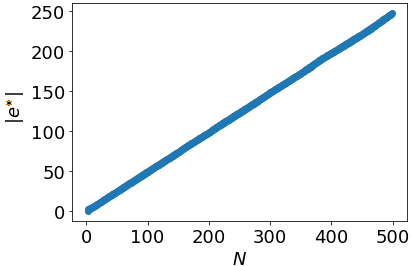}
    \caption{An estimate of $\e[|e^*|]$ as a function of the {number $N$ of nodes}, where $e^*$ is the first concordant hyperedge that we select and the initial opinions are normally distributed with a standard deviation of $\sigma = 1.2$. The confidence bound is $c = 1$. For each possible hyperedge size $n \in \{2, \ldots, N\}$, we run $10,000$ trials. For each trial, we randomly draw $n$ opinions from $\mathcal{N}(0, \sigma^2)$ and calculate the sample variance, which equals the discordance of those $n$ opinions. We approximate ${a_n} := \p[d(e, \bm{x}(0)) < c \mid |e| = n]$ by letting $\widehat{a_n}$ be the fraction of trials that result in a concordant set of $n$ opinions. For each $N \in \{2, \ldots, 500\}$, we have that $\e[|e^*|] \approx \frac{\sum_{n=2}^N n \widehat{a_n} {N \choose n}}{\sum_{n=2}^N \widehat{a_n} {N \choose n}}$.
    }
    \label{exp_conc_edge_size}
\end{figure}

Because $e^*$ is the first concordant hyperedge that we select, we update the nodes of $e^*$ to the opinion $\bar{x}_{e^*}(0)$, which is the mean of the initial opinions of the nodes that are incident to $e_0$. Let $\mu$ be the mean of the opinion distribution. If $e$ is an arbitrary hyperedge, then $\bar{x}_e(0) \to \mu$ in distribution as $|e| \to \infty$. It is necessarily also true that $\bar{x}_{e^*}(0) \to \mu$ in distribution as $|e^*| \to \infty$ because the sample mean and sample variance of a normal distribution are independent of each other (by Basu's theorem). Because $\e[|e^*|] \to \infty$ as $N \to \infty$, it follows that $\bar{x}_{e^*}(0) \to \mu$ in distribution as $N \to \infty$. In other words, we are updating the opinions of the nodes in $e^*$ to approximately {the value} $\mu$.

The observations above imply that (1) the first concordant hyperedge that we update includes a fraction of the nodes that is approximately constant as $N \to \infty$ (even for large $\sigma^2$) and (2) when we update all of those nodes, we are updating them to approximately {the value} $\mu$, which is the mean of the distribution. This decreases the total sample variance of {the opinions of the $N$ nodes} and increases the clustering of opinions near $\mu$, making it even more likely that the next hyperedge that we update will also be large and have a mean opinion that is centered near $\mu$. Eventually, the opinions converge to a consensus value {that is} near $\mu$.


\subsection{Our BCM on Hypergraphs with Community Structure}\label{sec:comm}

In this subsection, we examine our BCM on hypergraphs with planted community structure. 

Suppose that we partition the set of nodes in a hypergraph into communities and that each community has its own independent distribution of initial opinions. Specifically, we study our BCM on hypergraphs that we generate using the hypergraph stochastic block model (HSBM) of \cite{hsbm-paper}. An HSBM is a generative model for producing hypergraphs with community structure. Like a traditional SBM for ordinary graphs \cite{comm_review, comm_fortunato}, the probability that a hyperedge exists depends on the community memberships of its nodes. In this HSBM, the probability that a hyperedge exists also depends on the size of the hyperedge and on the number of nodes in the hypergraph. More precisely, let $\psi: V \to \{1,\ldots, k\}$ be a partition of the set of nodes into $k$ communities, where we assume without loss of generality that every community is non-empty. We denote the $i$th community by $C_i := \psi^{-1}(i)$. For $N = |V|$ and $n \in \{2, \ldots, N\}$, let $\alpha_{n, N} \in [0,1]$. For each $n \in \{2, \ldots, N\}$, let $B^n$ be a symmetric $k$-dimensional tensor of order $n$ whose entries take values in $[0,1]$. We generate a hypergraph as follows. For every subset $e = \{i_1, \ldots, i_n\} \in \mathcal{P}(V)$ of nodes, we include $e$ in the hypergraph with a probability of $\alpha_{n, N} B^n_{\psi(i_1), \ldots, \psi(i_n)}$.

In the simplest version of this HSBM, {each} intra-community hyperedge exists with an independent and uniform probability $p$ and {each} inter-community hyperedge exists with an independent and uniform probability $q$.

\begin{definition}
Consider the HSBM of \cite{hsbm-paper} with parameters $\alpha_{n,N} = 1$ for all $n, N$ and
\begin{equation*}
	B^n_{\psi(i_1), \ldots, \psi(i_n)} = 
\begin{cases}p \,, & \psi(i_1) = \cdots = \psi(i_n) \\
q \,, & \text{otherwise}
	\end{cases}
\end{equation*}
for some $p, q \in [0,1]$. We will refer to this HSBM as a $(p, q)$\emph{-HSBM}.
\end{definition}

{If} $q = 0$, the communities in {a} $(p, q)$-HSBM are disjoint and the opinions in {a} community cannot influence the opinions in other communities.

\begin{definition}
Let $\psi: V \to \{1, \ldots, k\}$ be a partition of the set of nodes into $k$ non-empty communities. The opinion state $\bm{x}^*$ is \textbf{polarized} if there are {opinions} $\gamma_1, \ldots, \gamma_k \in \mathbb{R}$, not all equal, such that $x_i^* = \gamma_{\psi(i)}$ for all $i$.
\end{definition}

{In} other words, $\bm{x}^*$ is polarized if each community is at consensus but the communities are not at consensus with each other. For example, if $q = 0$ and it is not the case that every community has the same initial mean opinion, then the opinion state converges to a limit state that {either is polarized or includes at least one} community whose nodes are not at consensus within the community.

{\begin{remark}
As we mentioned previously, some researchers refer to this situation as ``opinion fragmentation'' and reserve the term ``polarized'' for when there are exactly two opinion clusters in an opinion state.
\end{remark}}

The following theorem says that if $q \neq 0$, then the probability that the limit state of our BCM on a $(p,q)$-HSBM hypergraph is polarized approaches $0$ as $N \to \infty$.

\begin{theorem}\label{hsbm}
Suppose that {we generate $H$} from a $(p,q)$-HSBM with partition $\psi: V \to \{1, \ldots, k\}$ and that $q \neq 0$. Additionally, suppose that $c \neq 0$. Let $\bm{x}(0)$ be the initial opinion state, {where we draw $x_i(0)$ for each $i$ from a bounded distribution that is supported on $[a, b]$,} and let $\bm{x}(t)$ be the opinion state that is determined by \cref{eq:BCM_hypergraph_update}. Let $\bm{x}^*$ be the limit state, which exists by \cref{converge}. It then follows that the probability that $\bm{x}^*$ is polarized approaches $0$ as $N \to \infty$.
\end{theorem}

\begin{proof}
Because $\bm{x}^*$ is almost surely an absorbing state by \cref{no_updates}, it suffices to show that $\p[\bm{x}^* \text{ polarized and absorbing}] \to 0$ as $N \to \infty$. Suppose that $\bm{x}^*$ is a polarized, absorbing state. Because $\bm{x}^*$ is polarized, there are {opinions} $\gamma_1, \ldots, \gamma_k$ that are not all equal and that satisfy $x_i^* = \gamma_{\psi(i)}$. Without loss of generality, let $C_1$ be the largest community. It is necessarily true that $|C_1| \geq \lceil N/k \rceil$. Because $\{\gamma_i\}$ are not all equal, there is a community $j$ such that $\gamma_1 \neq \gamma_j$.

To find a contradiction, suppose that there is a hyperedge $e \in E$ of size $n > |\gamma_1 - \gamma_j|^2/c$ such that $e$ is incident to $n-1$ nodes in $C_1$ and {$1$} node in $C_j$. It follows that
\begin{equation*}
    0 < d(e, \bm{x}^*) = \frac{(\gamma_1 - \gamma_j)^2}{n} < \frac{(b-a)^2}{n} < c\,,
\end{equation*}
which contradicts the assumption that $\bm{x}^*$ is an absorbing state. As $N \to \infty$, the probability that there is no such $e \in E$ is
\begin{equation*}
    \lim_{N \to \infty} \p[\nexists \,e \in E]
    = \lim_{N \to \infty} (1-q)^{\sum_{i = \max\{0, 1 - n_0\}}^{|C_1| - n_0 + 1} {|C_1| \choose n_0 + i - 1} \times |C_j|}
    \leq \lim_{N \to \infty} (1-q)^{\lceil N/k \rceil}= 0 \,,
\end{equation*}
where $n_0 = \left\lfloor \frac{|\gamma_1 - \gamma_j|}{c} \right\rfloor$. 
\end{proof}

The following theorem says that if we also impose the condition $p = 1$ (so that each community forms a hyperclique), then the probability of reaching consensus approaches $1$ as $N \to \infty$ if all communities are sufficiently large.

\begin{theorem}\label{hsbm_p1}
Suppose that we generate a hypergraph $H$ from a $(p, q)$-HSBM with partition $\psi: V \to \{1, \ldots, k\}$ and that $p = 1$ and $q \neq 0$. Let $\bm{x}(0)$ be the initial opinion state, {where we draw $x_i(0)$ for each $i$ from a bounded distribution that is supported on $[a, b]$,} and let $\bm{x}(t)$ be the opinion state that is determined by \cref{eq:BCM_hypergraph_update}. Additionally, suppose that $|C_i| > (\frac{b-a}{\sqrt{c}} + 1)(\frac{(b-a)^2}{c}-1)$ for all $i \in \{1, \ldots, k\}$ and that $c \neq 0$. It then follows that $\p[\text{consensus}] \to 1$ as $N \to \infty$.
\end{theorem}

\begin{proof}
By \cref{converge}, the limit state $\bm{x}^* := \lim_{t \to \infty} \bm{x}(t)$ exists. By the same argument as in the proof of \cref{complete_converge}, the nodes in each community converge to consensus almost surely because $|C_i| > (\frac{b-a}{\sqrt{c}}+1)(\frac{(b-a)^2}{c} -1)$ for all $i \in \{1, \ldots, k\}$. That is, there exist {opinions} $\gamma_1, \ldots, \gamma_k \in [a,b]$ such that $x^*_i = \gamma_j$ for all $ i \in C_j$. If $\gamma_i = \gamma_j$ for all communities $i$ and $j$, then the opinion state converges to consensus. Otherwise, we assume without loss of generality that $C_1$ is the largest community and we let $C_j$ be a community such that $\gamma_j \neq \gamma_1$. Suppose that there is a hyperedge $ e \in E$ of size $n > \frac{(b-a)^2}{c}$ such that $e$ is incident to $n-1$ nodes in $C_1$ and {$1$} node in $C_j$. By the same argument as in the proof of \cref{hsbm}, the probability that such a hyperedge exists approaches $1$ as $N \to \infty$. {If} $e$ does exist, {then} $\bm{x}^*$ is not an absorbing state. By \cref{no_updates}, $\bm{x}^*$ is almost surely an absorbing state.
\end{proof}

In another version of the HSBM in \cite{hsbm-paper}, one requires that every inter-community hyperedge is small.

\begin{definition}
Consider the HSBM of \cite{hsbm-paper} with parameters $\alpha_{n, N} = 1$ for all $n, N$ and
\begin{equation*}
B^n_{\psi(i_1), \ldots, \psi(i_n)} = \begin{cases}
    p \,, & \psi(i_1) = \cdots = \psi(i_n) \\
    q \,, & \text{there exist } j \text{ and } k \text{ such that } \psi(i_j) \neq \psi(i_k) \text{ and } n \leq M \\
    0 \,, & \text{otherwise}
    \end{cases}
\end{equation*}
for some $M \geq 2$ and $p, q \in [0,1]$. We will refer to this HSBM as a $(p, q, M)$\emph{-HSBM}.
\end{definition}
For fixed $p, q, M$ and $N \to \infty$, the communities are ``almost'' disjoint; that is, the {number of inter-community hyperedges divided by the total number of hyperedges} approaches $0$. {The following theorem, which is useful to contrast with \cref{hsbm},} gives conditions under which a polarized opinion state is an absorbing state. \cref{thm_polarized} implies that echo chambers can form when all of the inter-community {hyperedges are sufficiently small.}

\begin{theorem}\label{thm_polarized}
Suppose that we generate a hypergraph $H$ from a $(p, q, M)$-HSBM with the partition $\psi: V \to \{1, \ldots, k\}$. Let $\gamma_1, \ldots, \gamma_k \in \mathbb{R}$, and let $\bm{x}^*$ be the polarized opinion state with $x_i^* = \gamma_{\psi(i)}$. If $\min_{i, j}\{(\gamma_i - \gamma_j)^2 \mid \gamma_i \neq \gamma_j\}/M > c$, then $\bm{x}^*$ is an absorbing state.
\end{theorem}

\begin{proof}
Without loss of generality, assume that $\gamma_i \neq \gamma_j$ if $i \neq j$. (If not, we combine any communities {with} the same opinion.) Let $e \in E$ be a hyperedge. Let $n = |e|$, and let $n_i = |\{j \in e \mid \psi(j) = i\}|$ be the number of nodes that are incident to $e$ and belong to community $i$. We have that
\begin{align}
    d(e, \bm{x}^*) &= \frac{1}{n-1}\Bigg(\sum_i n_i \Big(\gamma_i - \frac{1}{n}\sum_j n_j \gamma_j\Big)^2\Bigg) \notag \\
    &= \frac{1}{n-1} \Bigg(\sum_i n_i \gamma_i^2 - \frac{1}{n}\Big(\sum_i n_i \gamma_i\Big)^2\Bigg) \notag\\
    &= \frac{1}{n(n-1)} \Big(\sum_i n_i \gamma_i^2(n- n_i)
    - 2\sum_i \sum_{j < i}n_i n_j \gamma_i \gamma_j\Big) \notag\\
    &= \frac{1}{n(n-1)} \Big(\sum_i\sum_{j \neq i}n_in_j \gamma_i^2 -2\sum_i \sum_{j < i} n_i n_j \gamma_i \gamma_j\Big) \notag\\
    &= \frac{1}{n(n-1)}\sum_i \sum_{j < i} n_i n_j(\gamma_i - \gamma_j)^2 \,. \label{d_formula}
\end{align}
Let $\{i_*, j_*\} = \text{argmin}_{i, j} \{(\gamma_i - \gamma_j)^2 \mid \gamma_i \neq \gamma_j\}$. If there is an $\ell$ such that $n_{\ell} = n$, then it follows that $e \subseteq C_{\ell}$. Therefore, $x_i^* = \gamma_{\ell}$ for all $i \in e$, so $d(e, \bm{x}^*) = 0$. Otherwise, $n_{\ell} \neq n$ for all $\ell$. Fixing $n$ and enforcing the constraint that $n_{\ell} \neq n$ for all $\ell$, it follows that \cref{d_formula} is minimized either when 
\begin{equation*}
	n_{\ell} = \begin{cases} 
		0\,, & \ell \neq i_*, j_* \\ n-1\,, & \ell = i_* \\1\,, & \ell = j_*
			\end{cases}
\end{equation*}	
or when 
\begin{equation*}
	n_{\ell} = \begin{cases} 
				0\,, & \ell \neq i_*, j_* \\ 1\,, & \ell = i_* \\n-1\,, & \ell = j_*\,.
		\end{cases}
\end{equation*}
Therefore, $d(e, \bm{x}^*) \geq \frac{\min\{(\gamma_i - \gamma_j)^2 \mid \gamma_i \neq \gamma_j\}}{M} > c$ if $n_{\ell} \neq n$ for all $n$. {Therefore, for all} $e \in E$, it must be the case that either $d(e, \bm{x}^*) = 0$ or $d(e, \bm{x}^*) > c$.
\end{proof}

In \cref{polarized}, we show a typical simulation when the conditions of \cref{thm_polarized} are satisfied. In the limit state, the communities are polarized. All nodes in community $C_i$ converge to {the} opinion $\gamma_i = \sum_{j \in C_i} x_j(0)$, which is the mean of the initial opinions in $C_i$. By \cref{thm_polarized}, we know that this polarized opinion state is an absorbing state.

\begin{figure}
    \centering
    \includegraphics[width = .7\textwidth]{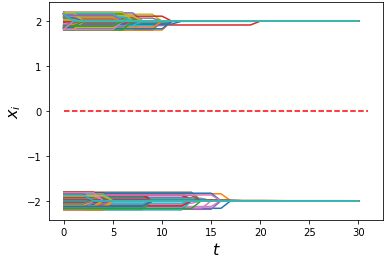}
    \caption{A typical simulation when the conditions of \cref{thm_polarized} are satisfied. There are $N = 1000$ nodes, which we assign to two equal-sized communities. We generate the hypergraph on which we run our BCM from a $(p, q, M)$-HSBM with $p = 1$, $q = 1$, and $M = 2$. The confidence bound is $c =1$. The initial opinions of the nodes in community $1$ are $x_i(0) \sim \mathcal{U}(1.8, 2.2)$, and the initial opinions of the nodes in community $2$ are $x_i(0) \sim \mathcal{U}(-2.2, -1.8)$. In the limit state, all nodes in community $1$ have opinion $\gamma_1 \approx 2$ and all nodes in community $2$ have opinion $\gamma_2 \approx -2$. This polarized opinion state is an absorbing state.}
    \label{polarized}
\end{figure}

\begin{remark}
By the same strategy as in the proofs of \cref{hsbm}, \cref{hsbm_p1}, and \cref{thm_polarized}, one can show that echo chambers do not usually form in a hypergraph that we generate using a $(p,q,M)$-HSBM if the initial opinions of the different communities are sufficiently close to each other. {In particular, if} we draw the initial opinions from a bounded interval $[a,b]$ with $(b-a)^2/M < c$ and if $|C_i| > (\frac{b-a}{\sqrt{c}})(\frac{(b-a)^2}{c} -1)$ for all $i$, then the probability of reaching consensus approaches $1$ as $N \to \infty$.

\end{remark}


\subsection{Our BCM on Sparse Hypergraphs}\label{sparse}

We now study our hypergraph BCM on sparse $G(N,\bm{m})$ hypergraphs. The $G(N, \bm{m})$ model is a generative hypergraph model that is defined analogously to the Erd\H{o}s--R\'{e}nyi $G(N, m)$ generative graph model. Each hypergraph that one constructs from the $G(N, \bm{m})$ model has $N$ nodes; for each possible hyperedge size $i \in \{2, \ldots, N\}$, we choose $m_i$ hyperedges of that size uniformly at random to include in the hypergraph.

In our Monte Carlo simulations, we set $N = 1000$ and $m_i = \max\{100, {N \choose i}\}$ for all $i$. We run $1000$ simulations with a confidence bound of $c = 1$ and initial opinions that we draw uniformly at random from $[-2, 2]$. In all trials, we find that the opinion state converges to consensus. We also run $1000$ trials with a confidence bound of $c = 1$ and initial opinions that we draw from the normal distribution with mean $\mu = 0$ and standard deviation $\sigma = 1.2$. All of these trials also converge to consensus. In \cref{sparse_plots}, we show a typical simulation for each of {these} two initial opinion distributions.

As we increase $N$ and increase the variance of the initial opinions, we observe that the time {that} it takes to converge increases but that the opinion state still converges to consensus.

\begin{figure}
    \centering
    \subfloat[$x_i(0) \sim \mathcal{U}(-2, 2)$]{\includegraphics[width = .5\textwidth]{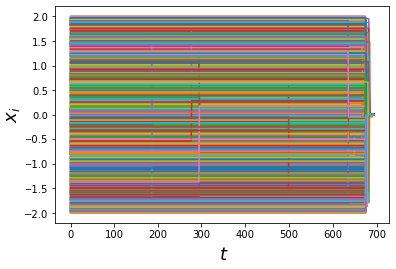}}
    \subfloat[$x_i(0) \sim \mathcal{N}(0, 1.2)$]{\includegraphics[width = .5\textwidth]{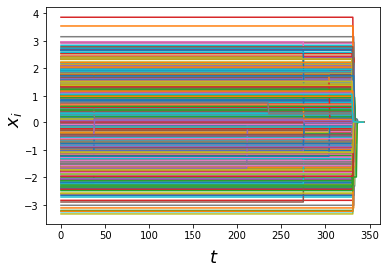}}
    \caption{Typical simulations on sparse $G(N, \bm{m})$ hypergraphs with a confidence bound of $c = 1$ and $N = 1000$ nodes that have (a) uniformly distributed initial opinions and (b) normally distributed initial opinions.
    Both of these simulations converge
    to consensus.}
    \label{sparse_plots}
\end{figure}


\subsection{Our BCM on an Enron E-mail Hypergraph}\label{sec:enron}

In \cref{sparse}, we studied our BCM on sparse hypergraphs. However, it is typically also the case that the hypergraphs that one constructs from empirical data are not merely sparse; they also have the property that their hyperedges are 
small in size in comparison to the number of nodes. As one example, we use a hypergraph that Benson et al. \cite{benson} constructed from the well-known (and infamous) Enron e-mail data set \cite{enron_original}. In this Enron e-mail hypergraph, nodes represent Enron employees and hyperedges represent e-mails between them. Each hyperedge is incident to the sender and recipients of one e-mail message. There are $N = 143$ nodes, but the maximum hyperedge size is only $18$. 

To examine our hypergraph BCM on the Enron e-mail hypergraph, we run $1000$ simulations with initial opinions that are uniformly distributed in $[0,1]$. In all $1000$ trials, the opinion state converges to {consensus}. In \protect\cref{enron_uniform}, we show the results of a typical simulation. We also run $1000$ simulations on the Enron e-mail hypergraph with initial opinions that are normally distributed with a mean of $\mu = 0$ and a variance of $\sigma^2 = 1$. The confidence bound is $c = 1$. In all trials, the opinion state converges to {consensus}. In \protect\cref{enron_normal}, we show the results of a typical simulation.

\begin{figure}
    \centering
    \subfloat[$x_i(0) \sim \mathcal{U}(0, 1)$]{\includegraphics[width=0.495\textwidth]{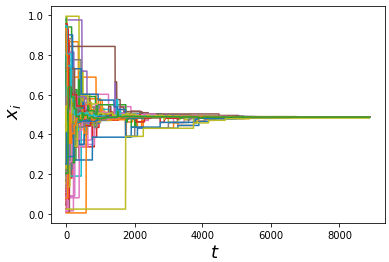} \label{enron_uniform}}
    \subfloat[$x_i(0) \sim \mathcal{N}(0, 1)$]{\includegraphics[width=0.495\textwidth]{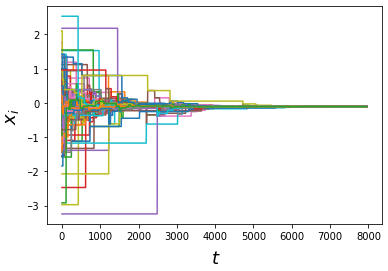}\label{enron_normal}}
    \caption{A typical simulation on the Enron e-mail hypergraph with a confidence bound of $c = 1$ and nodes with (a) uniformly distributed initial opinions and (b) normally distributed initial opinions. Both of these simulations converge to {consensus}.}
    \label{enron}
\end{figure}


\section{Convergence Time}\label{time} 

In this section, we analyze the convergence time of our hypergraph BCM.


\subsection{Conditions for Convergence in Finite Time}\label{finite_time}

We say that the opinion state \textit{converges in finite time} if there is a time $T$ such that $\bm{x}(T)$ is an absorbing state. At time $T$, no further opinion updates can occur. We use the following lemma to prove \cref{complete_finite_time}.

\begin{lemma}\label{primeSeq}
Let ${\tilde{e}} \in \mathcal{P}(V)$ be a subset of {the nodes in a hypergraph}. There is a finite sequence $\{{\tilde{e}_i}\}_{i=1}^m$ in $\mathcal{P}(V)$ such that (1) $|\tilde{e}_i|$ is prime for all $i \geq 1$ and (2) consecutively updating the nodes that are incident to $\tilde{e}_1, \ldots, \tilde{e}_m$ to their respective mean opinions would result in the same opinion state as updating the nodes that are incident to $\tilde{e}$ to {the mean of their opinions.}
\end{lemma}

\begin{proof}
Let $n = |\tilde{e}|$. If $n$ is prime, we are done. This also proves the base case $n = 2$. If $n$ is not prime, we can write $n = pm$, where $p$ is prime and $m < n$. Without loss of generality, {suppose that} $\tilde{e} = \{1,\ldots, n\}$. Updating the nodes of $\tilde{e}$ to the mean opinion in $\tilde{e}$ at time $t$ results in the opinion state 
\begin{equation}\label{eq:e0state}
    x_i= \begin{cases}
    \frac{1}{n} \sum_{i \in \tilde{e}} x_i(t) = \frac{1}{p} \sum_{k=0}^{p-1} \Big( \frac{1}{m} \sum_{j=1}^m x_{km + j}(t) \Big) \,, & i \leq n \\
    x_i(t) \,, & i > n\,.
    \end{cases}
\end{equation}
Updating the nodes of $\{km + 1, \ldots, km + m\} \in \mathcal{P}(V)$ for each $k \in \{0, \ldots, p-1\}$ to their respective mean opinions at time $t$ results in the opinion state
\begin{equation}\label{eq:intermediate_state}
    x_i = \begin{cases}
    	\frac{1}{m} \sum_{\ell = 1}^m x_{km+ \ell}(t)\,, & i = km +j \,, \,\, 1 \leq j \leq m \,, \,\,0 \leq k \leq p-1 \\
    x_i(t)\,, & i > n\,.
    \end{cases}
\end{equation}

By induction (because $m < n$), there is a sequence of elements of prime size in $\mathcal{P}(V)$ such that updating this sequence results in the same opinion state as updating $\{km + 1, \ldots, km + m\}$. Concatenating these $p$ sequences ({there is} one for each $k$) yields a sequence $\tilde{e}_1, \ldots, \tilde{e}_{\ell}$ of elements of prime size in $\mathcal{P}(V)$ such that updating the nodes of $\tilde{e}_1, \ldots, \tilde{e}_{\ell}$ to their respective mean opinions at time $t$ results in the opinion state \cref{eq:intermediate_state}. Updating the sequence $\{1, m + 1, \ldots, (p-1)m+1\}$, $\{2, m + 2, \ldots, (p-1)m + 2\}$, \ldots, $\{m, 2m, \ldots, pm\}$ of elements in $\mathcal{P}(V)$ then results in the opinion state \cref{eq:e0state}. Therefore, $\tilde{e}_1, \ldots, \tilde{e}_{\ell}, \{1, \ldots, (p-1)m+1\}, \{2, \ldots, (p-1)m + 2\}, \ldots, \{m, \ldots, pm\}$ is the desired sequence of prime-sized elements in $\mathcal{P}(V)$.
\end{proof}

\begin{theorem}\label{complete_finite_time}
Let $H$ be a hypergraph such that $\{ e \in \mathcal{P}(V) \mid |e| \text{ {is} prime}\} \subseteq E$. {Let $\bm{x}(0)$ be any initial opinion state, and let $\bm{x}(t)$ be the opinion state that is determined by \cref{eq:BCM_hypergraph_update}.} It follows that $\bm{x}(t)$ {almost surely converges in finite time.}
\end{theorem}

\begin{proof}
Let the time $t_0$, the matrices $K_1, \ldots, K_p$, and the node sets $I_1, \ldots, I_p$ be defined as in equation \cref{Alim}. This equation implies for all $ j, k \in I_i$ that
\begin{equation}\label{xlim}
	\lim_{t \to \infty} x_j(t) = \lim_{t \to \infty} x_k(t)\,.
\end{equation}
Updating the nodes of $I_1 \in \mathcal{P}(V)$ to {the mean of their opinions} results in consensus among the nodes of $I_1$. By \cref{primeSeq}, there is a sequence $\{e_i\}_{j=1}^n$ of prime-sized elements in $\mathcal{P}(V)$ such that consecutively updating $\{e_i\}$ results in the same opinion state as updating $I_1$. By hypothesis, the hypergraph $H$ includes $e_i$ for all $ i$ because $e_i$ is prime-sized. There is a time $t_1 \geq t_0$ such that $d(e_i, \bm{x}(t)) < c$ for all $ i$ for all times $t \geq t_1$. The probability of consecutively choosing the hyperedges of $\{e_i\}_{i=1}^n$ starting at a given time is $1/|E|^n > 0$, where $E$ is the hyperedge set of $H$. Because this probability is positive, {it is almost surely the case that the event of consecutively choosing these hyperedges occurs infinitely often}. Therefore, there is almost surely some time {after $t_1$ that we} consecutively choose the hyperedges $\{e_i\}$ for updating. Let $t_2 > t_1$ be the time that we choose the last hyperedge $e_n$ of the sequence. At this time, $x_j(t_2) = x_k(t_2)$ for all $ j, k \in I_1$. Similarly, we can find times $t_3 \leq \cdots \leq t_{p+1}$ for the nodes of $I_2, \ldots, I_p$. By equation \cref{Alim}, any hyperedge $e_t$ that we choose at $t \geq t_1$ is discordant if $e_t$ is not contained in some $I_i$. Therefore, for $t \geq t_{p+1}$, it follows that $d(e, \bm{x}(t)) = 0$ if there is an $i$ such that $e \subseteq I_i$; otherwise, $d(e, \bm{x}(t)) > c$. Therefore, $\bm{x}(t_{p+1})$ is an absorbing state.
\end{proof}

\begin{remark}
\cref{complete_finite_time} applies to {a} complete hypergraph.
\end{remark}

The following theorem gives a partial converse to \cref{complete_finite_time}.

\begin{theorem}\label{not_finite_time}
Let $H$ be a hypergraph with $N$ nodes and hyperedge set $E$. Suppose that there is a subset $e = \{i_1, \ldots, i_p\}\not\in E$ of nodes such that $e$ has prime size $p$, and suppose that the subhypergraph that is induced by $\{i_1, \ldots, i_p\}$ is connected. (The connectivity requirement implies that $p \geq 3$.) It then follows that there is an initial opinion state $\bm{x}(0)$ such that the opinion state $\bm{x}(t)$ {that is} determined by \cref{eq:BCM_hypergraph_update} {does} not converge in finite time.
\end{theorem}

\begin{proof}
Without loss of generality, $e = \{1, \ldots, p\}$. By the continuity of the discordance function, there is an $ \epsilon > 0$ such that if $x_1(0), \ldots, x_p(0) \in [0, \epsilon]$, then $d(e', \bm{x}(0)) < c$ for all $e' \subseteq e$. We choose $r$ such that $1/2^r < \epsilon$, and we let $x_1(0) = 0$ and $x_i(0) = 1/2^r$ for all $i \in \{2, \ldots, p\}$.
Additionally, there is an $M > \epsilon$ such that if $x_{p+1}(0), \ldots, x_N(0) > M$, then $d(e', \bm{x}(0)) > c $ for all $ e' \in E$ that satisfy $e' \not\subseteq \{1, \ldots, p\}$ and $e' \not\subseteq \{p+1, \ldots, n\}$. Let $x_{p+1}(0) = \cdots = x_N(0) = M.$ Because of these initial {values,} the nodes in $\{1, \ldots p\}$ never interact with the nodes in $\{p+1, \ldots, n\}$. Therefore, $x_1(t), \ldots, x_p(t) \in [0,\epsilon]$ for all $t$ and every $e' {\subseteq} e$ is concordant for all $ t$. Because the subhypergraph that is induced by the nodes $\{1, \ldots, p\}$ is connected, 
\begin{equation*}
	\lim_{t \to \infty} x_i(t) = \frac{1}{p} \sum_{j=1}^p x_j(0) = \frac{p-1}{2^rp}\, \,\text{  for all  }\,\, i \in \{1, \ldots, p\}\,.
\end{equation*}	
Let $e_t = \{i_1, \ldots, i_k\}\in E$ be the hyperedge that we choose at time $t$. If $e_t \not\subseteq e$, then $x_i(t+1) = x_i(t)$ for all $i \in e$. Otherwise, $e_t \subset e$ and for all $i \in e_t$, we have
\begin{equation*}
    x_i(t+1) = \frac{1}{k} \sum_{j=1}^k x_j(t) \,,
\end{equation*} 
{with $k \in \{2, \ldots, p-1\}$,} because $|e| = p$ and $e_t$ is a strict subset of $e$. By induction on $t$, there exist $a_{ij}(t) \in \mathbb{N}\cup \{0\}$ and $n_k(t) \in \mathbb{N} \cup \{0\}$ such that 
\begin{equation*}
	x_i(t+1) = \frac{1}{\prod_{k=2}^{p-1} k^{n_k(t)}} \sum_{j=1}^p a_{ij}(t) x_j(0) = \frac{1}{2^r\prod_{k=2}^{p-1} k^{n_k(t)}} \sum_{j=2}^p a_{ij}(t) \,\,\text{ for all }\,\, i \in e\,. 
\end{equation*}
For all $t$, we have
\begin{equation}\label{eq:finite_time}
	(p-1)2^r \prod_{k=2}^{p-1} k^{n_k(t)} \neq 2^r p \sum_{j=2}^p a_{ij}(t)
\end{equation}
because the left-hand side of \cref{eq:finite_time} is a non-zero integer without $p$ in its prime factorization and the right-hand side of \cref{eq:finite_time} is either $0$ (if $\sum a_{i_j}(t) = 0$) or a non-zero integer with $p$ in its prime factorization. Consequently, $x_i(t) \neq \lim_{t\to \infty} x_i(t)$ for all $t$ and all $i \in e$. Therefore, $\bm{x}(t)$ does not converge in finite time.
\end{proof}

Together, \cref{complete_finite_time} and \cref{not_finite_time} partially characterize the conditions for finite-time convergence of our BCM. If $H$ is a hypergraph whose hyperedge set includes all prime-sized subsets of nodes, then the opinion state {almost surely converges in finite time}. However, if the set of hyperedges does not include some prime-sized subset of nodes and the subhypergraph that is induced by those nodes is connected, then {it is not the case that the opinion state almost surely converges in finite time.} We do not have a characterization of the convergence time {when} this connectivity condition is not satisfied for any of the prime-sized elements of $\mathcal{P}(V)$ that are not in the {set of hyperedges.} However, we expect that ``most'' hypergraphs do not fall into this missing case. The connectivity condition is not hard to satisfy. For example, in the $G(N, \bm{m})$ model, we expect the subhypergraph that is induced by any {set $\{i_1, \ldots, i_p\}$ of nodes} to be connected whenever $p$ is sufficiently larger than the index of the first non-zero entry of $\bm{m}$. The vast majority of hypergraphs that are produced by the $G(N, \bm{m})$ model satisfy the conditions of either \cref{complete_finite_time} or \cref{not_finite_time}.


\subsection{A Phase Transition at $\sigma^2 = c$}

We now study the rate of convergence of the opinion state in our hypergraph BCM. We focus on {complete hypergraphs,} for which we {observe} {that there is} a phase transition in convergence time when the confidence bound is $c = \sigma^2$. {We prove that the convergence time grows at least exponentially fast with the number $N$ of nodes when $\sigma^2 > c$ and the initial opinions are normally distributed.} 

In \cref{fig:BCM_hyper_conv_time}, we simulate our BCM with $c = 1$ on {the} complete hypergraph with $50,000$ nodes and initial opinions that we seed independently by setting $x_i(0)\sim\mathcal{N}(0,\sigma^2)$ for $\sigma\in[0.9,1.1]$. We plot an empirical convergence time $t^*$, which we set to be the earliest time that the discordance function satisfies $d(e=V,\bm{x})<10^{-5}$. Our results are consistent with the existence of a phase transition in convergence time.

\begin{figure}
  \begin{center}\includegraphics[width=0.6\textwidth]{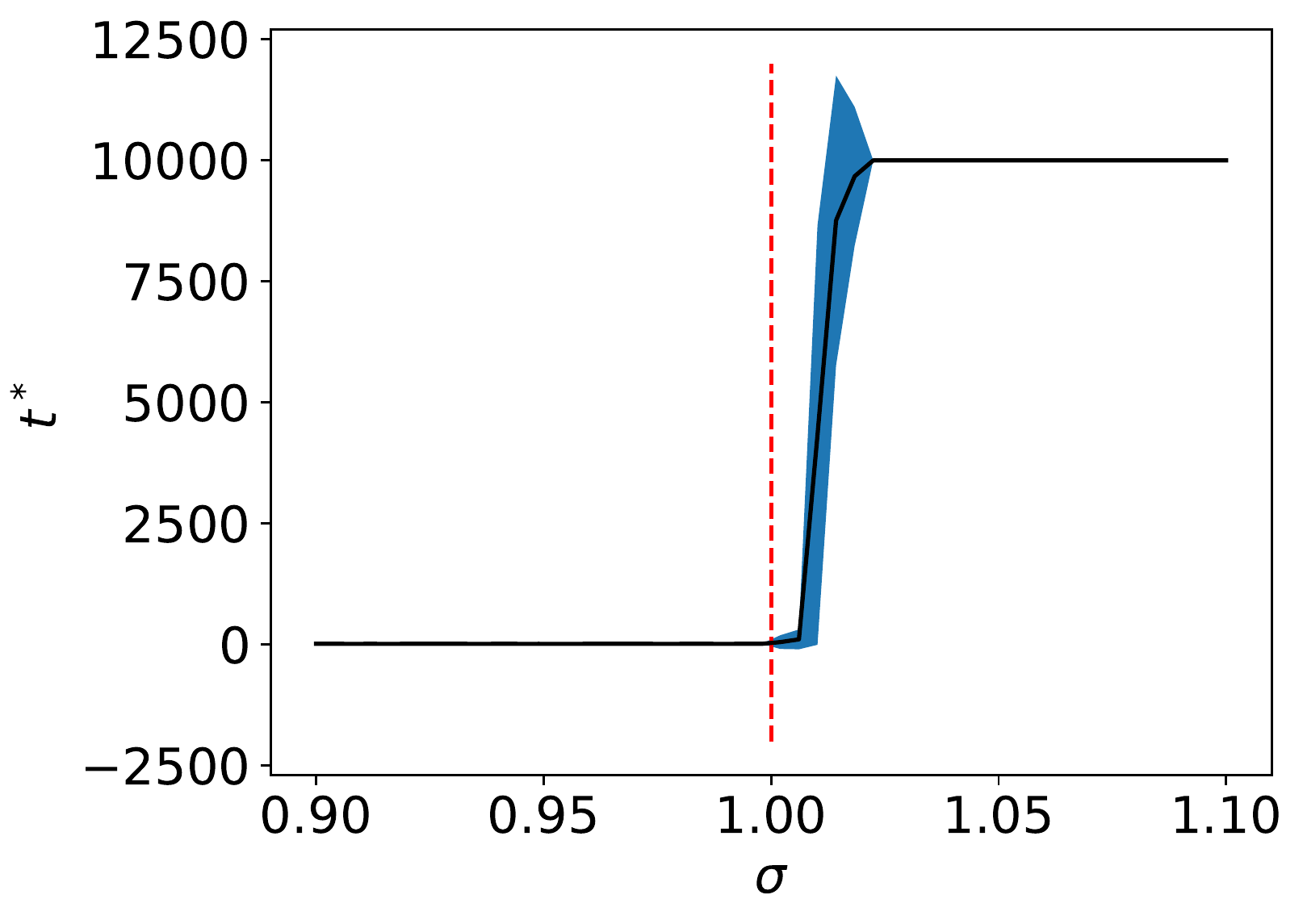}\end{center}
  \caption{Empirical convergence time {of} our BCM on {the} complete hypergraph with $50,000$ nodes and opinions that we seed independently using $x_i(0)\sim\mathcal{N}(0,\sigma^2)$. We consider all uniformly-spaced values $\sigma\in[0.9,1.1]$ with a step size of $\Delta \sigma = 0.004$. The empirical convergence time $t^*$ is the first time that the discordance function of the opinion state satisfies $d(V,\bm{x})<10^{-5}$. If the system does not reach such an opinion state by $t=10^4$, we record $t^*$ as $10^4$. We simulate 20 trials for each value of $\sigma$. The black curve gives the mean of $t^*$ over the trials, and the blue area depicts one standard deviation from the mean. We include a dashed red line at $\sigma^2=1 = c$ for reference.
  }
 \label{fig:BCM_hyper_conv_time}
\end{figure} 

In \cref{normal_time}, we prove that the convergence time grows at least exponentially fast as a function of $N$ if $c < \sigma^2$. The proof relies on \cref{edge_conc}, where we calculated the value of $\lim_{n \to \infty} \p[d(e, \bm{x}(0)) < c \mid |e| = n]$ for any initial opinion distribution with finite variance $\sigma^2$ and showed that there is a transition at $\sigma^2 = c$. We also need a bound on the convergence rate in this limit. {Using the inequalities that we derived in the proof of \cref{edge_conc} and} the fact that $\var[s_n^2] = O(\frac{1}{n})${, we see} that the convergence rate is $O(\frac{1}{n})$ whenever $\sigma^2 \neq c$. When the initial opinions are normally distributed, we can derive a much tighter bound on the convergence rate.

\begin{lemma}\label{normal_edge_rate}
Suppose that we draw the initial opinions from a normal distribution with variance $\sigma^2 \neq c$, and let $\lambda = \frac{c}{\sigma^2}$. It follows that
\begin{align*}
    \p[d(e, \bm{x}(0)) < c \mid |e| = n] &\geq 1 - \Big (e^{\frac{1}{2}(1- \lambda)} \sqrt{\lambda}\Big)^{n-1} \,, \qquad{\lambda > 1} \,, \\
    \p[d(e, \bm{x}(0)) < c \mid |e| = n] &\leq \Big( e^{\frac{1}{2}(1 - \lambda)}\sqrt{\lambda}\Big)^{n-1} \,, \qquad{\lambda < 1}\,.
\end{align*}
Therefore, $\p[d(e, \bm{x}(0)) < c \mid |e| = n]$ converges exponentially fast as $n \to \infty$.
\end{lemma}

\begin{proof}
The discordance of a hyperedge $e$ at time $0$ is the sample variance of the opinions $\{x_j(0) \mid j \in e\}$. Let $s_n^2$ denote the sample variance of $n$ opinions. We have that $\p[d(e, \bm{x}(0)) < c \mid |e| = n] = \p[s_n^2 < c]$.

\underline{Case 1}. Suppose that $\lambda > 1$. By Chernoff's bound \cite{degroot_stats},
\begin{equation}\label{lam_big}
    \p[s_n^2 \geq c] = \p\Big[ \frac{n-1}{\sigma^2} s_n^2 \geq \lambda (n-1)\Big] \leq \frac{\e[e^{t \frac{n-1}{\sigma^2} s_n^2}]}{e^{t\lambda(n-1)}} \, \,\text{ for all }\,\, t > 0\,.
\end{equation}
By Cochran's {theorem~\cite{knight}, $\frac{n-1}{\sigma^2}s_n^2 \sim \chi_{n-1}^2$.}
When $t < \frac{1}{2}$, we have that
\begin{align*}
    \e[e^{t \frac{n-1}{\sigma^2} s_n^2}] &= \frac{1}{2^{(n-1)/2}\Gamma(\frac{n-1}{2})}\int_0^{\infty} e^{tx} x^{\frac{n-1}{2}-1} e^{-x/2}dx \\
    &= \frac{1}{\Gamma(\frac{n-1}{2})} \int_0^{\infty} x^{\frac{n-1}{2}-1} e^{-x(1-2t)} dx \\
    &= \frac{1}{\Gamma(\frac{n-1}{2})} (1-2t)^{-\frac{n-1}{2}}\int_0^{\infty} x^{\frac{n-1}{2}-1} e^{-x} dx \\
    &= (1-2t)^{-\frac{n-1}{2}}.
\end{align*}
{Therefore,} when $0 < t< \frac{1}{2}$, \cref{lam_big} becomes
\begin{equation}\label{lam_big2}
    \p[s_n^2 \geq c] \leq \Big(\frac{1}{e^{t\lambda}\sqrt{1 - 2t} }\Big)^{n-1}\,.
\end{equation}
Setting $t = \frac{1}{2}(1 - \frac{1}{\lambda})$ in equation  \cref{lam_big2} yields
\begin{equation*}
    \p[s_n^2 \geq c] \leq \Big (e^{\frac{1}{2}(1 - \lambda)}\sqrt{\lambda} \Big)^{n-1}.
\end{equation*}

\vspace{5mm}

\underline{Case 2}. Suppose that $\lambda < 1$. By Chernoff's bound, 
\begin{equation}\label{lam_small}
    \p[s_n^2 < c ] = \p\Big[\frac{n-1}{\sigma^2} s_n^2 < \lambda(n-1)\Big] \leq \frac{\e[e^{-t \frac{n-1}{\sigma^2}s_n^2}]}{e^{-t\lambda(n-1)}} \, \,\text{ for all }\,\, t > 0\,.
\end{equation}
Similarly to Case 1, we compute that
\begin{equation*}
	\e[e^{-t \frac{n-1}{\sigma^2}s_n^2}] = \Big(\frac{1}{1 + 2t}\Big)^{\frac{n-1}{2}} \, \,\text{ for all }\,\, t > 0\,.
\end{equation*}	
Therefore, when $t > 0$, \cref{lam_small} becomes
\begin{equation}\label{lam_small2}
    \p[s_n^2 < c] \leq \Bigg( \frac{e^{\lambda t}}{\sqrt{1 + 2t}}\Bigg)^{n-1}\,.
\end{equation}
Setting $t = \frac{1}{2}(\frac{1}{\lambda} - 1)$ in \cref{lam_small2} yields
\begin{equation*}
	\p[s_n^2 < c] \leq \Big( e^{\frac{1}{2}(1 - \lambda)}\sqrt{\lambda}\Big)^{n-1}.
\end{equation*}
\end{proof}
\begin{remark}
When $\sigma^2 = c$ and the initial opinions are normally distributed, we have numerical evidence that $\p[d(e, \bm{x}(0)) < c \mid |e| = n ]$ converges to $\frac{1}{2}$ exponentially fast as $n \to \infty$, but we do not have a mathematical proof of the convergence rate.
\end{remark}

{We say that the opinion state of our hypergraph BCM} \textit{converges with threshold $\epsilon$ at time $T_{\epsilon}$} if{,} for all $ e \in E$, either $d(e, \bm{x}(T_{\epsilon})) \leq \epsilon$ or $d(e, \bm{x}(T_{\epsilon})) \geq c$. When $\epsilon = 0$, the time $T_{\epsilon}$ is exactly the convergence time.

\begin{theorem}\label{normal_time}
Let $H$ be the complete hypergraph with $N$ nodes. Let $\bm{x}(0)$ be the initial opinion state and suppose that $x_i(0) \sim \mathcal{N}(\mu, \sigma^2)$, where $\sigma^2 > c$ and $c \neq 0$. Finally, let $\epsilon < c$ and let $T_{\epsilon, N}$ be the convergence time with threshold $\epsilon$. It then follows that 
\begin{equation*}
	\e[T_{\epsilon, N}] = \Omega\Big(r \Big[\frac{2}{r+1}\Big]^N\Big)\,,
\end{equation*}	
where $r = e^{\frac{1}{2}(1 - c/\sigma^2)}\sqrt{\frac{c}{\sigma^2}}< 1$.
\end{theorem}

\begin{proof}
Let $A_N$ be the Bernoulli random variable that equals $1$ if there is a hyperedge $e \in E$ such that $\epsilon < d(e, \bm{x}(0)) < c$ and equals $0$ if there is no such hyperedge. We have that
\begin{equation*}
    \e[T_{\epsilon, N}] = \e[T_{\epsilon, N} \mid A_N = 1]\p[A_N = 1] + \e[T_{\epsilon, N} \mid A_N = 0]\p[A_N = 0]\,.
\end{equation*}
If $A_N = 0$, then $T_{\epsilon, N} = 0$. As $N \to \infty$, we have
\begin{equation*}
    \lim_{N \to \infty} \p[A_N = 0] = \lim_{N \to \infty}\prod_{n =2}^N (1 - \p[\epsilon < d(e, \bm{x}(0)) < c \mid |e| = n])^{N \choose n} = 0\,,
\end{equation*}
where $e$ denotes a hyperedge that we choose uniformly at random. Let $s_N$ be the first time that we select a concordant hyperedge. If $A_N = 1$, then $T_{\epsilon, N} \geq s_N$. 
Let $X_N$ be the fraction of hyperedges in $H$ that are concordant at time $0$. We calculate
\begin{align*}
    \e[T_{\epsilon, N}] &\geq \e[T_{\epsilon, N} \mid A_N = 1]\p[A_N = 1] \\
    &\geq \e[s_N]\p[A_N = 1] \\
    &= \frac{\p[A_N = 1]}{\e[X_N]} \,, \\
    \e[X_N] &= \frac{\sum_{n=2}^N {N \choose n} \p[d(e, \bm{x}(0)) < c \mid |e| = n]}{2^N - N - 1} 
    \\
    &\leq \frac{\sum_{n=2}^N {N \choose n}r^{n-1}}{2^N - N - 1} \qquad{(\text{by \cref{normal_edge_rate}})} 
    \\
    &= \frac{\frac{1}{r}( (r+1)^N - N r -1)}{2^N - N - 1} \,.
\end{align*}
Therefore, as $N \to \infty$, we obtain
\begin{equation*}
    \e[T_{\epsilon, N}] \geq r \p[A_N = 1] \frac{2^N - N-1}{(r+1)^N - Nr -1}\sim r\Big(\frac{2}{r+1}\Big)^N\,.
\end{equation*}
\end{proof}
\begin{remark}
{\cref{normal_time}} applies to the convergence time $T_{0, N}$, which is almost surely finite by \cref{complete_finite_time}.
\end{remark}


\section{Opinion Jumping}\label{sec:jump}

We now study ``opinion jumping'', a phenomenon that occurs in our hypergraph BCM that cannot occur in standard dyadic BCMs. An \textit{opinion jump} occurs at time $t$ if there is a node $i$ such that $|x_i(t+1) - x_i(t)| > c$. The number of opinion jumps that occur at time $t$ is the number of nodes $i$ that satisfy $|x_i(t+1) - x_i(t)| > c$.

An opinion jump can occur only if the size of the selected hyperedge is at least $3$. Therefore, this behavior requires polyadic interactions; it cannot occur on BCMs on ordinary graphs. Moreover, we believe that it is one of the driving behaviors that causes our hypergraph BCM to converge to consensus so much more easily than is the case for standard dyadic BCMs. For examples of opinion jumping, see \Cref{complete_normal} and \cref{enron}. In this section, we quantify how common it is for an opinion jump to occur.

\begin{lemma}\label{J_eq}
Let $J_t$ be the number of opinion jumps that occur at time $t$. Suppose that the distribution of opinions at time $t$ has a mean of $\mu$ and a variance of $\sigma^2 < \infty$. Let $p_n = \p[|x_i - \bar{x}_e| > c \mid i \in e\,,\, |e| = n\,,\, d(e, \bm{x}(t)) < c]$ be the probability that a node's opinion is farther than $c$ from the mean opinion of the nodes in a {concordant} size-$n$ hyperedge that is incident to the node, and let $p= \p[|x- \mu| > c]$. Let $a_n = \p[d(e, \bm{x}(t)) < c \mid |e| = n]$ be the probability that a size-$n$ hyperedge is concordant, and let $a = \lim_{n \to \infty} a_n$ be the limiting probability of concordance. Finally, let $e_t$ be the hyperedge that we select at time $t$. {The expected number of opinion jumps is}
\begin{align*}
    \e[J_t] = \Bigg(pa \e[|e_t|] &+ p \sum_{n=2}^N (a_n - a) \p[|e_t| = n]n + \sum_{n=2}^{N} (p_n - p)a \p[|e_t| = n] n \\
    &+ \sum_{n = 2}^N (p_n - p)(a_n - a) \p[|e_t| = n]n\Bigg)\,.
\end{align*}
\end{lemma}

\begin{remark}
The quantities $p_n$, $p$, $a_n$, and $a$ depend on the distribution of opinions at time $t$. The value of $a$ is given by \cref{eq:limprobconc}.
\end{remark}

\begin{proof}
For $j \geq 1$, we have
\begin{align*}
    \p[J_t = j] &= \sum_{n=2}^N \p[J_t = j \,\text{ and }\, |e_t| = n] \\
    &= \sum_n \p[J_t = j \mid |e_t| = n]\p[|e_t| = n] \\
    &= \sum_n \p[J_t = j \,\text{ and }\, d(e_t, \bm{x}(t)) < c \mid |e_t| = n]\p[|e_t| = n] \qquad{(\text{because } j \geq 1)}\\
    &= \sum_n \p[J_t = j \mid d(e_t, \bm{x}(t)) < c, |e_t| = n]a_n \p[|e_t| = n] \\
    &= \sum_n {n \choose j} p_n^j (1-p_n)^j a_n \p[|e_t| = n]\,.
    \end{align*}
    Therefore,
    \begin{align*}
    \mathbb{E}[J_t] &= \sum_{j=0}^N \p[J_t = j]j \notag \\
    &=\sum_j\sum_n j{n \choose j}p_n^j (1-p_n)^j a_n \p[|e_t| = n] \notag \\
    &= \sum_n a_n\p[|e_t| = n] \sum_j j {n \choose j} p_n^j (1-p_n)^j \notag \\
    &= \sum_n a_n \p[|e_t| = n] p_nn \label{bernoulli} \\
    &= \Bigg( pa \e[|e_t|] + p \sum_{n=2}^N (a_n - a) \p[|e_t| = n]n + \sum_{n=2}^{N} (p_n - p)a \p[|e_t| = n] n \notag \\
    &\qquad + \sum_{n = 2}^N (p_n - p)(a_n - a) \p[|e_t| = n]n \Bigg) \notag \,.
\end{align*}
\end{proof}

{We} use \cref{J_eq} to derive the asymptotic behavior of $\e[J_0]$. The following proposition says that, under certain conditions, $\e[J_0]$ grows linearly with the mean hyperedge size of the hypergraph on which our BCM occurs.

\begin{proposition}\label{J_asym}
{Let $\{H_m\}$ be a sequence of hypergraphs, with the associated sequence $\{V_m\}$ of nodes and sequence $\{E_m\}$ of hyperedge sets. Let $\{g_m\}$, where
\begin{equation*}
    g_m(n) = \frac{|\{e \in E_m \mid |e| = n\}|}{|E_m|}\,,
\end{equation*}
be the corresponding sequence of hyperedge-size distributions.} Suppose that we draw the initial opinions from the same distribution for all $H_m$, and let {$p$, $p_n$, $a$, and $a_n$} be defined as in \cref{J_eq}. Finally, let $J_0^m$ be the number of opinion jumps that occur at time $0$ for $H_m$. {Suppose that} $(a - a_n)n \to 0$ as $n \to \infty$, $g_m(n) \to 0$ for all $ n$ as $m \to \infty${,} and $a(p_n-p)n \to 0$ as $n \to \infty${. It then follows that} $\e[J_0^m] \sim pa \e[|e_0|]$ as $m \to \infty$.
\end{proposition}

\begin{remark}
{When} $H_m$ is the complete hypergraph with $m$ nodes, we have that $g_m(n) \to 0$ for all $n$ as $m \to \infty$. The values of $a$, $a_n$, $p$, and $p_n$ depend only on the opinion distribution at time $t$. The value of $a$ is given by \cref{eq:limprobconc}. If the initial opinions are normally distributed, then \cref{normal_edge_rate} implies that $(a_n -a)n \to 0$ when $\sigma^2 \neq c$. The exact value of $p$ depends on the initial distribution, but it tends to increase with $\sigma^2$. Our numerical computations suggest that ${a}(p_n - p)n \to 0$ when the initial opinion distribution is normally distributed with variance $\sigma^2 \neq c$.
\end{remark}

\begin{proof}
\cref{J_eq} implies that
\begin{align*}
     \e[J_0^m] &= pa\e[|e_0|] + a \sum_n (p_n-p)n g_m(n) + p \sum_n (a_n -a)ng_m(n) \\
    &\qquad + \sum_n (p-p_n)(a_n-a)n g_m(n)\,.
\end{align*}
Let $x_n$ be any sequence such that $x_n \to 0$. For any $m$, the quantity $\sum_n x_n g_m(n)$ is a weighted average of $\{x_n\}$. As $m \to \infty$, the weights concentrate at larger values of $n$. Therefore, because $x_n \to 0$, it follows that $\sum_n x_n g_m(n) \to 0$ as $m \to \infty$.

We apply the above argument to $x_n = {a}(p_n -p)n$, $x_n = (a_n- a)n$, and $x_n = (p-p_n)(a_n-a)n$ to prove the proposition.
\end{proof}

In Figure~\ref{jumps}, we present numerical results that support the claim that $\e[J_0] \approx pa \e[|e_0|]$ when the {hyperedge-size} distribution concentrates at large {hyperedge} sizes. We generate hypergraphs from the $G(N, \bm{m})$ hypergraph model for different values of $\bm{m}$. For each hypergraph, we run {$500$} trials of our hypergraph BCM and record the mean value of $J_0$ for the hypergraph. We plot the mean value of $J_0$ versus the mean hyperedge size $\e[|e_0|]$ in the hypergraph. We show results for initial opinions that are normally distributed with standard deviations of $\sigma = 0.6$, $\sigma = 0.8$, $\sigma =1$, and $\sigma = 1.2$. We use a confidence bound of $c = 1$ in all trials. The claim $\e[J_0] \approx pa \e[|e_0|]$ implies that for each $\sigma$, there should be a linear relationship with a slope of $pa$, where $p$ and $a$ depend on $\sigma$. Whenever $\sigma^2 < c$, the limiting probability of concordance is $a = 1$. The slope when {$\sigma = 0.8$} is steeper than the slope when {$\sigma = 0.6$} because $p$ becomes larger for progressively larger {values of $\sigma$}. When $\sigma = 1$, the value of $p$ is larger than for {$\sigma =  0.6$ and $\sigma = 0.8$.} {However,} the limiting probability of concordance is only $a = \frac{1}{2}$ {and} the slope $pa$ {for the case $\sigma = 1$} is slightly less steep than {the slope} when {$\sigma = 0.8$}. We observe that the linear relationship between $\e[J_0]$ and $\e[|e_0|]$ is not as strong when $\sigma^2 = c = 1$ as when $\sigma^2 \neq c$. Based on numerical evidence, we suspect that this is because $(p_n -p)n \not\to 0$ when $\sigma^2 = c$. Finally, when $\sigma = 1.2$, we {observe} that $\e[J_0] \approx 0$ because the limiting probability of concordance is $a = 0$ whenever $\sigma^2 > c$.

\begin{figure}
    \centering
    \includegraphics[width = .7\textwidth]{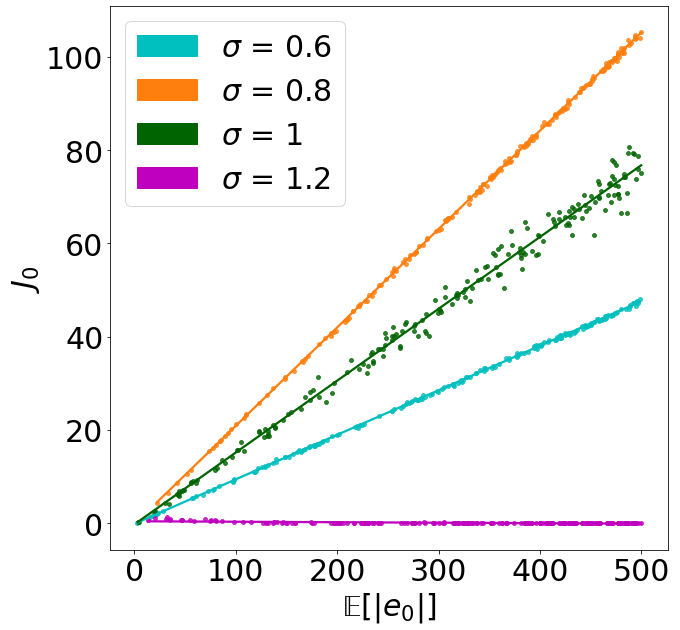}
    \caption{Empirical evidence for the linear relationship between $\e[J_0]$ and $\e[|e_0|]$, where $J_0$ is the number of opinion jumps that occur at time $0$ and the hyperedge $e_0$ is the hyperedge that we choose uniformly at random at time $0$. In {this} numerical experiment, the initial opinions are normally distributed with a mean of $0$ and standard deviation of $\sigma$, which takes values of $0.6$, $0.8$, $1$, and $1.2$. For each $\sigma$, we generate $200$ hypergraphs with $N = 1000$ nodes. To construct the $\ell$th hypergraph, we choose $x_{\ell} \in [0,1]$ uniformly at random and set $m^{(\ell)}_n = {N \choose i}x_{\ell}^n$. We generate the $\ell$th hypergraph from the $G(N, \bm{m}^{(\ell)})$ model, which has $N$ nodes and $m_n^{(\ell)}$ hyperedges of size $n$ that we choose uniformly at random. For each hypergraph, we run one step of our hypergraph BCM with confidence bound $c = 1$ and record the number $J_0$ of opinion jumps. We run $500$ trials {for} each hypergraph; {for each trial,} we preserve the hypergraph structure, reset the initial opinions, and run one step of our hypergraph BCM. We record the mean value of $J_0$ over these $500$ trials{,} and we plot it versus the mean hyperedge size $\e[|e_0|]$ in the hypergraph. Each data point corresponds to the trial results for a single hypergraph for a given value of $\sigma$. For each $\sigma$, we plot the line {of best fit.}
    }
    \label{jumps}
\end{figure}


\section{Conclusions and Discussion}\label{conclusion}

We formulated a bounded-confidence model (BCM) {on} hypergraphs and explored its properties {using} both mathematical analysis and Monte Carlo simulations. We showed that polyadic (i.e., ``higher-order'') interactions play an important role in opinion dynamics and that one cannot reduce such interactions to pairwise interactions on a graph. In our hypergraph BCM, we also demonstrated a novel phenomenon, which we called ``opinion jumping'', that requires polyadic interactions {to manifest.} Therefore, opinion jumping cannot occur in standard dyadic BCMs.

We proved that our hypergraph BCM converges to consensus on {complete hypergraphs} for a wide variety of initial conditions. This is very different from what occurs in standard dyadic BCMs, which usually converge to multiple opinion clusters. We also studied the effects of a variety of initial opinion distributions on the dynamics of our BCM. In particular, we examined the convergence properties of our BCM when the initial opinion distribution is bounded (but not necessarily uniform), normally-distributed, or has a variance $\sigma^2$ that is less than the confidence bound $c$. Based on our results, we expect that the limit states of dyadic BCMs also depend on the initial opinion distribution (although{, to the best of our knowledge,} this is not something that has been studied in detail in prior research) and that the number of opinion clusters depends not only on the confidence bound $c$ but also on the relative sizes of $c$ and the variance $\sigma^2$ of the initial opinion distribution.

We also explored the dependence of the limit state of our hypergraph BCM on community structure. We proved that the opinion state can become polarized if the intra-community hyperedges are sufficiently small in size. This leads to the formation of echo chambers. {We also showed} that if the intra-community hyperedges are unbounded in size and if the communities are {sufficiently large} and form hypercliques, then the opinion state converges to consensus.

We {demonstrated} that there is a phase transition in convergence time {of our BCM} on {complete hypergraphs} when the confidence bound $c$ equals the variance $\sigma^2$ of the initial opinion distribution. When $\sigma^2 > c$, the convergence time of our BCM on {a} complete hypergraph depends {at least} exponentially on the number $N$ of nodes. This has implications for the feasibility of using Monte Carlo simulations for simulating our BCM on {a} complete hypergraph when $\sigma^2 > c$ and $N$ is large (and, more generally, on any hypergraph {in which} large hyperedges {constitute} a significant proportion of the hyperedge set, because large hyperedges are likely to be discordant and it thus takes many time steps to choose a concordant hyperedge). It is fascinating that there is a phase transition in convergence {time but} not in the limit state. By contrast, in the standard dyadic DW model, there is a phase transition in convergence time at the same {confidence-bound} threshold $c^*$ at which there is a phase transition in the limit state \cite{meng}. We also proved that our hypergraph BCM converges in finite time on {complete hypergraphs}. This is similar to what occurs in the standard dyadic HK model, which converges in finite time on complete graphs; however, it differs from the DW model, which tends not to converge in finite time on complete graphs.

Because of opinion jumping, which requires polyadic interactions, nodes with extreme opinions can move quickly {towards the mean opinion in our hypergraph BCM}. {When the variance $\sigma^2$ satisfies} $\sigma^2 < c$, we showed that the number of opinion jumps in the first time step grows roughly linearly with the mean hyperedge size in a hypergraph and that it becomes larger for progressively larger values of $\sigma^2$ up to the value $c$. {It will be worthwhile to determine} the precise necessary conditions for opinion jumping in hypergraph BCMs.

In our work, we made several modeling choices, and there are numerous alternatives that are also worth studying. For example, one can formulate a hypergraph BCM that uses synchronous updates of opinions instead of asynchronous updates. For example, at each discrete time, suppose that each node updates its opinion to the mean of the {incident} hyperedges' mean opinions.{\footnote{The model of Sahasrabuddhe et al.~\cite{sahas} is related to this idea. In their model, the amount of influence of a hyperedge depends on the opinion state of the system.}} We believe that such a synchronous model {has} similar limit states as our asynchronous model, but we expect such models to converge much more quickly to a limit state. One can also develop synchronous models in which each node updates its opinion to a weighted mean of the {incident} hyperedges' mean opinions. Such heterogeneity models a situation in which some friendship groups {exert} more influence on a person than others. This extends the notion of trust from the dyadic DeGroot model \cite{degroot}. Another of our {modeling} choices was our {discordance function. Instead of choosing} $d = d_1$ for the discordance function, it is worthwhile to study the entire family of discordance functions $d_{\alpha}$ for $\alpha \in [0,1]$ that we defined in \cref{eq:BCM_discordance}. The case $\alpha = 0$ is particularly interesting because it models a scenario in which it is more difficult for large groups of people to agree than it is for small groups. Another variation of our model involves incorporating heterogeneous confidence bounds, which models {situations} in which some individuals are {persuaded more easily} than others.

There are a variety of other avenues to explore. For example, it is worth conducting a deeper investigation of the role of hypergraph topology on the limit states of hypergraph BCMs, and one can also study BCMs on simplicial complexes (which {entail} various constraints on {which polyadic interactions are permissible}). We believe that the presence of large hyperedges that connect some subset of a hypergraph's nodes will facilitate the convergence of those nodes to {consensus}. One can also develop adaptive (i.e., coevolving) hypergraph BCMs, such as by modifying the hypergraph structure at each {time step} in response to the current opinion state. For example, one can allow agents to strategically rewire in a way that maximizes their influence or perhaps to simply leave a hyperedge when the other nodes that are incident to it become too {``annoying''} ({which can occur sometimes} in discussion groups on social media). From a control-theoretic perspective, one can {examine} how much control the {nodes in a hypergraph} (or an outside controller) can have in steering an opinion state towards a particular limit state by choosing which hyperedges to update or rewire.

\appendix

\section{Continuum Formalism}

Instead of running Monte Carlo simulations, which are costly, one can study the ``continuum'' formalism of Ref.~\cite{ben2003bifurcations} using numerical integration. Consider a hypergraph in which every hyperedge is of size $\ell \in L {\subseteq} \{2,\ldots,n\}$, and let $P(x,t)\,dx$ be the probability density function that indicates how many nodes have opinions in the interval $[x,x+dx]$ at time $t$. The distribution $P(x,t)$ evolves according to the rate equation
{{\tiny
\begin{equation}\label{eq:BCM_hyper_rate_eq}
	\frac{\partial}{\partial t} P(x,t) = \sum_{\ell \in L} \int_{\{\sum_{j=1}^\ell (y_j - \bar{y})^2 < c(\ell-1)\} }  dy_1 \cdots dy_\ell P(y_1,t)\times \cdots \times P(y_\ell,t) \left[\delta(x-\bar{y} )- \delta (x-y_1)\right]\,.
\end{equation}}
}
The $\ell$-fold integrals in the summand are over all $\ell$-tuples of points whose sample variance is less than $c$. The delta functions reflect the gains (from nodes that update their opinion to $\bar{y}$) and losses (from nodes that update their opinions and thus change their current opinion) in the update process. We do not study \cref{eq:BCM_hyper_rate_eq} in the present paper, but it seems interesting to {examine} in future work.


\section*{Acknowledgments}

We thank Phil Chodrow and Ryan Wilkinson for helpful discussions and comments.




\end{document}

%% file: shared.tex
\usepackage{lipsum}
\usepackage{amsfonts}
\usepackage{amssymb}
\usepackage{graphicx}
\usepackage{subfig}
\usepackage{bm}
\usepackage{epstopdf}
\usepackage{algorithmic}
\usepackage{cite}
\usepackage{hyperref}
\ifpdf
  \DeclareGraphicsExtensions{.eps,.pdf,.png,.jpg}
\else
  \DeclareGraphicsExtensions{.eps}
\fi

\newcommand{\p}{\mathbb{P}}
\newcommand{\e}{\mathbb{E}}
\newcommand{\var}{\text{Var}}

\definecolor{colorA}{rgb}{0.0 0.0, 1.0}

\newsiamremark{remark}{Remark}
\newsiamremark{hypothesis}{Hypothesis}
\crefname{hypothesis}{Hypothesis}{Hypotheses}
\newsiamthm{claim}{Claim}

\headers{A Bounded-Confidence Model of Opinion Dynamics on Hypergraphs}{A. Hickok, Y. Kureh, H. Z. Brooks, M. Feng, M. A. Porter}

\title{A Bounded-Confidence Model of Opinion Dynamics on Hypergraphs
\thanks{
\funding{AH, MAP, MF, and YK acknowledge support from the National Science Foundation (Grant
No. 1922952) through the Algorithms for Threat Detection (ATD) program.}}}

\author{Abigail Hickok\thanks{Department of Mathematics, University of California, Los Angeles, CA, USA
  (\email{ahickok@math.ucla.edu}, \email{jkureh@gmail.com}).}
\and Yacoub Kureh\footnotemark[2]
\and Heather Z. Brooks\thanks{Department of Mathematics, Harvey Mudd College, Claremont, CA, USA (\email{hzinnbrooks@g.hmc.edu}).}
\and Michelle Feng\thanks{Department of Computing + Mathematical Sciences, California Institute of Technology, Pasadena, CA, USA (\email{mfeng@caltech.edu}).}
\and Mason A. Porter\thanks{Department of Mathematics, University of California, Los Angeles, CA, USA; Santa Fe Institute, Santa Fe, NM, USA (\email{mason@math.ucla.edu}).}}

\usepackage{amsopn}

\makeatletter
\newcommand*{\addFileDependency}[1]{
  \typeout{(#1)}
  \@addtofilelist{#1}
  \IfFileExists{#1}{}{\typeout{No file #1.}}
}
\makeatother